\documentclass[conference]{asl}
\usepackage{rotate}

\newcommand{\setE}              { {\mathcal E} }

\newcommand{\setG}              { {\mathcal G} }

\newcommand{\setI}              { {\mathcal I} }

\newcommand{\setM}              { {\mathcal M} }

\newcommand{\setP}              { {\mathcal P} }

\newcommand{\Nat}               { {\tt N}  }
\newcommand{\nth}               { {\tt nth}  }

\newcommand{\PA}                { {\tt PA}}

\newcommand{\false}             { {\tt false} }
\newcommand{\False}             { {\tt False} }
\newcommand{\true }             { {\tt true } }
\newcommand{\True }             { {\tt True } }

\newcommand{\interpret} [1]     { {[\![#1]\!]} }

\newcommand{\Par}               { {\tt Par} }

\newcommand{\List}            { {\tt List} }
\newcommand{\FV}              { {\tt FV} }

\newcommand{\comp}            { {\mbox{\tiny $\circ$}} }

\newcommand{\dom}             { {\tt dom} }

\newcommand{\Var}             { {\tt Var} }

\newcommand{\inj}             { {\tt j} }
\newcommand{\EM}              { {\tt EM} }

\newcommand{\Bool}            { {\tt Bool} }

\newcommand{\JUST}            { {\tt JUST} }

\newcommand{\Term}            { {\tt Term} }
\newcommand{\Player}          { {\tt P} }
\newcommand{\Opponent}        { {\tt O} }
\newcommand{\turno}           { {\tt t} }
\newcommand{\vince}           { {\tt w} }
\newcommand{\Foglie}          { {\tt Lv} }
\newcommand{\Mosse}           { {\tt Mv} }
\newcommand{\nil}             { {\tt nil} }

\newcommand{\STOP}            { {\tt STOP} }

\newcommand{\DROP}            { {\tt DROP} }
\newcommand{\END}             { {\tt END} }
\newcommand{\EXH}             { {\tt EXH} }
\newcommand{\exh}             { {\tt exh} }
\newcommand{\REPEAT}          { {\tt REPEAT} }
\newcommand{\SKIP}            { {\tt SKIP} }
\newcommand{\TRYALL}          { {\tt TRYALL} }
\newcommand{\ch}              { {\tt ch} }
\newcommand{\Ch}              { {\tt Ch} }
\newcommand{\EXISTSGAME}      { {\bigvee} }
\newcommand{\FORALLGAME}      { {\bigwedge} }
\newcommand{\EXISTS}          { {\exists \!\!\! \exists} }
\newcommand{\FORALL}          { {\forall \!\!\! \forall} }
\newcommand{\Vaananen}        { \mbox{V\"{a}\"{a}n\"{a}nen} }
\newcommand{\Fraisse}         { \mbox{Fra\"{i}ss\'{e}} }
\newtheorem{lemma}{Lemma}
\newtheorem{definition}{Definition}
\newtheorem{theorem}{Theorem}
\newtheorem{corollary}{Corollary}

\def\confname{}
\def\copyrightyear{2016}

\title{A Sound, Complete and Effective Second Order Game Semantics}

\author{Stefano Berardi}

\twoaddress{Dipartimento di Informatica\\ Universit\`{a} di Torino\\
Corso Svizzera 185, 10149 Torino, Italy} 
{Dipartimento di Informatica\\ Universit\`{a} di Torino\\
Corso Svizzera 185, 10149 Torino, Italy}
{3}

\email{stefano@di.unito.it}

\urladdr{http://www.di.unito.it/~stefano/}

\thanks{}

\begin{document}

\maketitle

\begin{abstract}
We define a game semantics for second order classical arithmetic $\PA^2$ (with quantification over predicates on integers and full comprehension axiom). Our semantics is effective: moves are described by a finite amount of information and whenever there is some winning strategy for the player defending the truth of the formula, then there is some primitive recursive winning strategy. Then we show that our game semantics is sound and complete for the truth assignment for formulas of $\PA^2$. In our game model, the value of a predicate variable is some family of ``generic'' games. This value is ``unknown'' during the play, but at the end of the play it is used by a ``judge of the play'' to decide who is the winner.
\end{abstract}

\section{Introduction}
\label{section-introduction}
Let us denote Second Order Classical Arithmetic with $\PA^2$. By this we mean: arithmetic with quantification over natural numbers and sets of natural numbers, with \emph{full} comprehension axiom for sets of natural numbers, or, alternatively, with the elimination rule $\FORALL X.A \implies A[P/X]$, for any predicate $P$ of the language. 

Our long-term goal is the following proof-theoretical analysis of $\PA^2$: we want to provide a characterization for the primitive recursive relations $R$ which are provably well-founded in $\PA^2$. For this reason, we are looking for an \emph{effective} game semantics: whenever there is some winning strategy for the first player, we want to have some primitive recursive winning strategy, and all moves should be described by a finite amount of information. We do not allow a move selecting an infinite set, but we allow a move selecting a finite description of an infinite set. 


This paper is a first step toward this goal. We define an effective second order game semantics for sequents of $\PA^2$. We interpret any sequent $\Gamma = ( A_0, \ldots, A_{n-1} )$ of $\PA^2$ by some game $\interpret{\Gamma}$. We plan to use this interpretation in another paper, in order to provide the required characterization. 

No effective game semantics exists yet for a logic corresponding to full Second Order Arithmetic. There are non-effective game semantics for logical systems equivalent to the $\Sigma^1_1$, $\Sigma^1_2$-fragments of Second Order Logic, and even for a logic equivalent to Second Order Logic. Indeed, Independence Friendly Logic has a game semantics and it is equivalent to the $\Sigma^1_1$-fragment of Second Order Logic (see Mann, Sandu and Sevenster \cite{Sandu}, Thm. 6.10, 6.16). Independence Friendly Fixpoints has a complete semantics with parity games defined by Bradfield (\cite{Bradfield-2003}): since the winner of a parity game is defined by an arithmetical formula, completeness of the semantics implies that the formulas of Independence Friendly Fixpoints are expressible by $\Sigma^1_2$-formulas. The reverse is proved in \cite{Bradfield-2003}, Thm. 18. It is also known that Independence Friendly Fixpoints may express full second order logic, but on \emph{finite} structures only (\cite{Bradfield-2005}). 
$\Vaananen$ \cite{Vaananen} used Ehrenfeucht-$\Fraisse$ game semantics to model Team Logic, a logic of implicit functional dependency equivalent to Second Order Logic. Implicitly, the game semantics of $\Vaananen$ defines a game semantics for Second Order Logic. However, all these game semantics are non-effective: there are formulas with winning strategies but no primitive recursive winning strategy, and in the case of Team Logic, a move may select \emph{any} element of a structure, including an infinite set. Besides, in these logical systems there is no explicit notion of quantification over a predicate or function variable.

De Lataillade defined a game semantics for system $F$, a second order functional language (\cite{JD08a, JD08b}), in which there an explicit notion of quantification over a formula variable, and which is effective in the sense that all moves are described by a finite amount of information. However, there is no obvious way to adapt his semantics to Second Order Logic or Arithmetic. First, De Lataillade provides no notion of winner: having a winner is not relevant for his goal of studying a functional language, but it is essential in order to interpret logic and arithmetic. Second, De Lataillade interprets a quantification $\FORALL X.A[X]$ on formulas as the fact that we may defend $A[P]$, for any formula $P$ of second order propositional logic. Again, this choice is suitable for studying a second order functional language, but if we transfer it to logic and arithmetic, it clashes with the fact that in the standard model of second order arithmetic there are sets of integers which are definable by no second order predicate $P$. In a game semantics of this kind, if we are able to convince our opponent that $A[P]$ is true for all predicates $P$ definable in second order arithmetic, then, in order to convince him that $\FORALL X.A[X]$ is true, we have to ask him to believe the statement: \emph{``if $A[P]$, for any predicate $P$ definable in second order arithmetic, then $\FORALL X.A[X]$''}. But our opponent may refuse to believe it: this statement is true in the model of Set Theory consisting of all constructible sets (\cite{Addison}), but it is false in some other model of Set Theory (\cite{Martin-Solovay}, Thm. 3).

The main contribution  of this paper is defining a game semantics in which we convince our opponent that $\FORALL X.A[X]$ is true by convincing him that $A[X]$ is true, for a family $\{X(n) | n \in \Nat\}$ of ``generic'' games, indexed over the set $\Nat$ of natural numbers. We represent a ``generic'' game $X(n)$ by a game having zero moves, whose winner is fixed but ``unknown'' during the play. The only information available during the play is that the games interpreting $X(n)$ and the negation of $X(n)$ are dual games, if we lose one we win the other. At the end of the play, some ``judge of the play'' announces the winner of each game $X(n)$, and uses it to decide the winner of the play, but no player may use the value of the parameter $X$ to decide his moves, in the much the same way in Independence-Friendly Logic no player may use the value of a ``slashed'' quantifier variable in order to choose the value of a quantifier variable. Our interpretation of an atomic formula $X(n)$ for a predicate variable $X$ is similar to the interpretation of an atomic formula $p(n)$ for a predicate constant $p$ by Lorenzen (\cite{Lorenzen}, \S 1, Def. (D10)): in both cases there is no way of discussing $X(n)$ or $p(n)$, all we may do is to match affirmed or negated occurrences of $X(n)$ or of $p(n)$, and claim that we are able to win one of them. The difference is that we interpret in this way all sub-formulas $\FORALL X.A$ occurring in a sequent $\Gamma$, while Lorenzen interpreted in this way the whole sequent $\Gamma$. Lorenzen defined an interpretation for First Order Intuitionistic Logic, and did not have to interpret $\FORALL X.A$.

 
Our main result is that our second order game semantics is sound, complete, and \emph{effective} for $\PA^2$. 
If we drop the effectiveness condition, we may obtain a simple second order game semantics which is sound and complete for truth, by generalizing Tarski games from first order quantifiers 
to second order quantifiers. In this case we ask that, in order to defend the truth of $\FORALL X.A[X]$ (for $X$ unary predicate variable), we should be able to defend the truth of $A[\psi]$, for any value $\psi:\Nat \rightarrow \Bool$ assigned to $X$. This is no effective semantics: the move selecting $\psi$ requires an infinite information ($\psi$ has domain $\Nat$), and besides there are formulas whose Tarski game has a winning strategy but no recursive winning strategy. 

The effective game semantics of $\PA^2$ we introduce is new 
and conceptually simple: for this reason, we think that it is interesting of its own right. ``Conceptually simple'', however, does not mean ``'logically simple'': by Tarski's undefinability theorem, any interpretation of the truth for $\PA^2$ cannot be defined in $\PA^2$ itself. 
This is the plan of the paper. In \S \ref{section-game} we propose our game interpretation for second order quantifiers. In \S \ref{section-pa} we introduce second order classical arithmetic $\PA^2$, the standard notion of truth for it, and our game semantics for it. In \S \ref{section-completeness} we prove that our game semantics is sound, complete and effective for $\PA^2$. In \S \ref{section-conclusion} we compare our second order game semantics with the second order game semantics for system $F$ by De Lataillade (\cite{JD08a, JD08b}).

\section{A notion of game for second order arithmetical formulas}
\label{section-game}
In this section we introduce a notion of ``game with parameters'', denoting a family of set theoretical games. Our first step is to precise very carefully how we represent at most countable trees.

%
%


\subsection{Representing at most countable trees}
Let $\setI$ be any countable set. In this subsection we introduce lists, tree supports and tree structures over some countable set $\setI$. We could always assume that $\setI$ is the set $\Nat$ of natural numbers, but we prefer a more abstract approach. 

\emph{Lists}. Let $\List(\setI)$ denote the set of all finite lists over $\setI$ and $\List_\infty(\setI)$ denote the set of all infinite lists over $\setI$. $\List_{\le \infty}(\setI) = \List(\setI) \cup \List_\infty(\setI)$ is the set of finite or infinite lists on $\setI$. We denote the finite or infinite list with elements $n_0, n_1, n_2, \ldots$ by $\langle n_0, n_1, n_2, \ldots \rangle$. We call $\nil = \langle \rangle$ the empty list, and $\langle n \rangle$ the one-element list. 

If $ x = \langle n_0, \ldots, n_{k-1} \rangle \in \List(\setI)$, $y = \langle m_0, \ldots, m_{h-1}, \ldots \rangle \in \List_{\le \infty}(\setI)$, we set $x @ y = \langle n_0, \ldots, n_{k-1}, \ m_0, \ldots, m_{h-1}, \ldots \rangle  \in \List_{\le \infty}(\setI)$, and we call $x@y$ the concatenation of $x,y$. $x@y$ is finite or infinite according if $y$ is finite or infinite. We extend concatenation to the case in which one or both lists are replaced by elements of $\setI$: if $i,j \in \setI$ we set  $x @ i = x @ \langle i \rangle$, $j @ y = \langle j \rangle @ y$ and $i @ j = \langle i \rangle @ \langle j \rangle$. We denote by $\le$ the prefix order on $\List_{\le \infty}(\setI)$, and we define $<_1$
by $x <_1 x @ i$ for all $x \in \List(\setI)$, $i \in \setI$. 

We distinguish between a ``tree support'', which is a ``plain'' tree, and a ``tree structure'', which is a tree with some additional information.

 \emph{Tree support}. A tree support over $\setI$, a \emph{tree support} for short, is any set $I \subseteq \List(\setI)$ such that  $\nil \in I$ and $I$ is closed under prefix. We call $\nil$ the root of $I$, any $x \in I$ a node of $I$, and any $x @ n \in I$ a child of $x$ in $I$ of index $n \in \setI$. A leaf of $I$ is any node of $I$ with no children in $I$. We write $\Foglie(I) \subseteq I$ for the subset of leaves of $I$. We call $I = \{\nil\}$ the support for the atomic tree (the one-node tree). A branch of $I$ is any (possibly infinite) sequence $\langle n_0, n_1, n_2, \ldots \rangle \in \List_{\le \infty}(\setI)$ over $\setI$, whose finite prefixes are all in $\setI$. In our formalism, the finite branches of $I$ are exactly the elements $x \in I$: $x$ is used to represent the branch of $I$ whose last node is $x$. We write $I_\infty$ for the set of infinite branches of $I$. We set $I_{\le \infty} = I \cup I_\infty$ for the set of branches of $I$. Given any node $x \in I$, the support $I_x$ for the sub-tree of descendants of $x$ in $I$ is defined by: $I_x = \{y \in \List(\setI) | x@y \in I\}$. By definition, $I_\nil = I$. If $\langle i \rangle \in I$, we call $I_{\langle i \rangle}$ an immediate sub-tree of $I$. If $I \subseteq \List(\setI)$ is tree support, and $i \in \setI$, we set $i @ I = \{i @ x | x \in I\}$. 


The canonical injection $\inj_x: (I_x)_{\le \infty} \rightarrow I_{\le \infty}$ is defined by $\inj_x(y) = x@y$ for all $y \in (I_x)_{\le \infty}$. $\inj_x(y)$ is finite or infinite according if $y$ is. 
If $y \in I_x$ is finite, we think of $\inj_x(y)$ as the node in $I$ representing $y$. By definition unfolding, we have $I_x = \inj_x^{-1}(I)$. $y$ is a leaf in $I_x$ if and only if $\inj_x(y)$ is a leaf in $I$: thus, $\Foglie(I_x) = \inj_x^{-1}(\Foglie(I))$.


\emph{Tree structure}. A tree structure over $\setI$, a tree structure for short, is a list $T = (|T|, \Ch_T, \ch_T )$, with $|T|$ tree support and $\Ch_T:|T| \rightarrow [0, \omega]$ computing the number of children of each $x \in T$, and $\ch_T : |T|, \Nat \rightarrow T$ returning the $n$-th child $x @ i_n $ of any $x \in T$ whenever $n < \Ch(x)$, returning $\nil$ o.w.. A tree structure $T$ is \emph{primitive recursive} if $|T|$ is a primitive recursive subset of $\List(\setI)$ and $\Ch_T, \ch_T$ are primitive recursive maps. 

Any tree support has a unique expansion to a tree structure, but the two concepts are different: a tree support may be primitive recursive and yet its tree structure may not be recursive, in the case $\ch_T$ is not recursive. If $T$ is a primitive recursive tree structure, then some frequently used predicate and functions on $T$ are primitive recursive. For instance: $\Foglie(|T|)$ is a primitive recursive predicate, because we assumed having a primitive recursive map computing the number of children of any node, and the leaves are the nodes with $0$ children. 

Given a tree structure $(|T|, \Ch_T, \ch_T)$, and a node $x \in |T|$, the sub-tree structure $T_x$ of descendants of $x$ in $T$ is defined by assigning to each node of $|T|_x$ the same values we assign to its image in $|T|$. To put otherwise, we set: $T_x = (|T|_x, \Ch_T \comp \inj_x,  \ch_T \comp \inj_x)$.

\subsection{Parametric Games}
In this sub-section we recall the folk-lore notion of set-theoretical game, and we extend it to the notion of parametric game, which we use to represent formulas having free predicate variables. 

We assume we have two players $\{ \Player, \Opponent \}$. We call $\Player$ ``Player'' and $\Opponent$ ``Opponent''. If $g \in \{ \Player, \Opponent \}$, then $g^\bot$ denotes the player opposed to $g$: we set $\Player^\bot=\Opponent$ and $\Opponent^\bot=\Player$. $(.)^\bot$ is involutory without a fixed point: $g^{\bot\bot} = g$ and $g^{\bot} \not = g$ for all $g \in \{ \Player, \Opponent \}$. We fix some countable set $\setM$, which we call the set of ``moves''.

A set-theoretical game includes a tree support over $\setM$. The nodes of the tree are called the \emph{positions} of the game and are lists of moves, the root is the initial position and is the empty list. There is a rule deciding, for each position, which player should move next. The player moving next selects a child of the current position as next position. As a result the players define a branch of the tree, which either ends in a leaf of the tree, or it is infinite. There is a rule of the game deciding who is the winner both for plays ending in a leaf and for infinite plays.

A parametric game is a slightly more general notion: the winner of some leaf $x$ of the game $G$ may be not settled by the definition of $G$, but may be a parameter $a$, whose range is $\{\Player, \Opponent\}$. A leaf $x$ with parameter $a$ represents a ``generic'' game, whose winner is decided by some ``judge of the play'', but it is not announced before the play is over. Indeed, the value of $a$ is not part of the definition of a parametric game, and therefore cannot be used by a strategy to decide the next move, even if the winner of a play may depend on the choice of $a$. We use the leaf $x$ in order to represent an atomic formula $X(n)$ which may be instanced to true or false.  


We suppose to be fixed some countable set $\Par$ of parameters having an involutory operation without fixed point $a \in \Par \mapsto a^\bot \in \Par$, representing negation. 
If $I \subseteq \Par$ then $I^\bot = \{ a^\bot | a \in I \}$. We say that $I$ is \emph{self-dual} if $I^\bot = I$. The smallest self-dual set including $I$ is $J = I \cup I^\bot$. Indeed, if $I \subseteq J$ and $J$ is self-dual, then $I^\bot \subseteq J^\bot = J$, hence $I \cup I^\bot \subseteq J$; and $(I \cup I^\bot)^\bot = I^\bot \cup I^{\bot\bot} = (I^\bot \cup I)$. 


We may now define parametric games.

\begin{definition}[Parametric Games]
 A parametric game structure over $\setM$, a \emph{parametric game} for short is any list $G = (|G|, \Ch_G, \ch_G, \turno_G, \vince_G, G_{\Player}, G_{\Opponent}) $, consisting of:
\begin{enumerate}
\item
some tree structure $(|G|, \Ch_G, \ch_G)$ over $\setM$, which we call the tree of positions of the game.
\item
some map $\turno_G: |G|\setminus \Foglie(|G|) \rightarrow \{ \Player, \Opponent \}$, taking any position $x \in |G|\setminus \Foglie(|G|) $ and returning the player $g=\turno_G(x)$ moving from $x \in |G|\setminus \Foglie(|G|)$. 
\item
some map $\vince_G: \Foglie(|G|)  \rightarrow \{ \Player, \Opponent \} \cup \Par $, taking any position $x \in \Foglie(|G|)$ and returning either the player $g = \vince_G(x) \in \Par$ winning a play which ends in $x$, or some parameter $a = \vince_G(x) \in \Par$.
\item
some partition $( G_{\Player}, G_{\Opponent} )$ over the set $|G|_{\infty}$ of infinite branches of $|G|$ among those won by $\Player$ and those won by $\Opponent$.
\end{enumerate}
We say that $G$ is a \emph{set-theoretical game} if there are no parameters in $G$, that is, if $\vince_G(\Foglie(|G|)  \subseteq \{ \Player, \Opponent \}$. $\setP$ is the set of parametric games and $\setG \subseteq \setP$ is the subset of set-theoretical games. 
\end{definition}

We use the traditional game 
terminology.

\begin{definition}[Game terminology]
\begin{enumerate}
\item
Any $p \in |G|$ is a position of $G$.
\item
Any $p \in |G|_{\le \infty}$ is a play of $G$.
\item
If $m \in \Nat$ and $x@m \in |G|$ them $m$ is a move from $x$ in $G$. 
\item
$\Mosse(G) = \{m \in \Nat | \exists x \in |G|.(x@m \in |G|)\}$ is the set of moves from some $x \in |G|$.
\item
An initial move of $G$ is a move from the root $\nil$ of $G$ 
\end{enumerate}
\end{definition}
 
We use a parametric game $G \in \setP$ in order to represent a family of set-theoretical games $\rho(G) \in \setG$ depending on a game assignment $\rho$ to the parameters labeling some leaves of $G$. $\rho$ plays the role of judge of the play, deciding the winner when a game ends in a leaf with a parameter. Therefore our game semantics requires a notion of parameter assignment $\rho$. In the case $G$ interprets some second order formula $A[X]$, $\rho$ corresponds to an interpretation of $X$ by some boolean function $\psi$, therefore to some predicate assignment to the predicate variable $X$ of $A[X]$.

\begin{definition}[Parameter assignment]
Assume $G \in \setP$ and $I, J \subseteq \Par$ are self-dual.
\begin{enumerate}
\item
The \emph{set of parameters} of $G$ is $\FV(G) =$ the smallest self-dual set including $ \vince_G(\Foglie(|G|))\cap \Par$, the set  of parameters assigned to some leaf of $G$. 
\item
An $I$-assignment is any dual-preserving map $\rho: I \rightarrow \{ \Player, \Opponent \}$ (i.e., such that $\rho(a^\bot) = \rho(a)^\bot$ for all $a, a^\bot \in I$). A $G$-assignment is any $I$-assignment for some $I \supseteq \FV(G)$.
\item
$\setE(I)$ is the set of all $I$-assignments and $\setE(G)$ of all $G$-assignments (of all $\setE(I)$-assignment for some $I \supseteq \FV(G)$). 
\end{enumerate}
We extend any assignment $\rho$ to a map $:\Par \cup \{\Player, \Opponent\} \rightarrow \{\Player, \Opponent\}$ by $\rho(g) = g$ for all $g \in \{\Player, \Opponent\}$.
\end{definition}

If $\rho:I \rightarrow \{ \Player, \Opponent \}$, $\eta : J \rightarrow \{ \Player, \Opponent \}$, then we define the over-writing $\rho, \eta: I \cup J \rightarrow \{ \Player, \Opponent \}$ by $(\rho, \eta)(a) = \eta(a)$ for all $a \in J$ and $(\rho, \eta)(a) = \rho(a)$ for all $a \in I \setminus J$.

If $\FV(G) = \emptyset$, then $\rho = \emptyset$ is a $G$-assignment and $\rho(G) = G$. 


Given $G \in \setP$, we define a set-theoretical play $\rho(G) \in \setG$. $\rho(G)$ is obtained by replacing any parameter $a$ in $G$ with the player $\rho(a) \in \{ \Player, \Opponent \}$. 


\begin{definition} [$\rho(G)$]
Assume that $G = (|G|, \Ch_G, \ch_G, \turno_G, \vince_G, G_{\Player}, G_{\Opponent}) \in \setP$ is any parametric game. Assume $\rho \in \setE(G)$ is any $G$-assignment. 
%
Then we set:
 $\rho(G) = (|G|, \Ch_G, \ch_G, \turno_G, \rho \comp \vince_{G}, G_{\Player}, G_{\Opponent})$ 
\end{definition}


Since $|\rho(G)| = |G|$, then the strategies for $G$ and for all $\rho(G)$ are the same: this is another way of expressing the fact that a strategy cannot use the values of the parameters of the game to decide the next move. However, the same strategy may be $\Player$-winning for some $\rho_1(G)$ and not $\Player$-winning for some $\rho_2(G)$, because the winner in a leaf of $\rho(G)$ may depend on $\rho$.

We call $ (|G|, \Ch_G, \ch_G, \turno_G, \vince_G) $ the \emph{finitary part} of $G$ and $(G_{\Player}, G_{\Opponent})$ the \emph{infinitary part} of $G$. $G$ has a \emph{primitive recursive finitary part} if the predicates and functions of the finitary part of $G$ are all primitive recursive. For instance, $\rho(G)$ is primitive recursive if both the finitary part of $G$ and the map $\rho$ are primitive recursive. 

Predicates and functions of the finitary part of $G$ have domain the countable set $\setM$ of moves, or some list over $\setM$: this is why we call them ``finitary''. $(G_{\Player}, G_{\Opponent})$ is a partition over a set of infinite lists over $\setM$, hence a predicate over infinite objects: this is why we call it ``infinitary''. When $G$ has primitive recursive finitary part, its infinitary part $(G_\Player,G_\Opponent)$ may still be a non-computable predicate. In this paper we interpret the truth of $L(\PA^2)$ by games with primitive recursive finitary part, the effective part of our semantics. We will prove that our semantics is sound and complete: as a corollary, by Tarski's undefinability theorem, the set of winning conditions $(G_{\Player}, G_{\Opponent})$ we use cannot be defined in $L(\PA^2)$.

For any $x \in |G|$, the sub-game $G_x$ of descendants of $x$ in $G$ is defined by assigning to each node and to each infinite branch of $G_x$ the same values we assign to its image in $G$. The dual game $G^\bot$ is defined by switching the role of $\Player$ and $\Opponent$, and $a$ with $a^\bot$ for any $a \in \FV(G)$. If $f:I \rightarrow \{\Player, \Opponent\}$ is any map, we define the dual map $f^\bot$ point-wise, by $f^\bot(i) = f(i)^\bot$ for all $i \in I$. 

\begin{definition}[Sub-games and dual games]
Let $G \in \setP$
and $x \in |G|$ 
\begin{enumerate}
\item
$G_x = (|G|_x, \Ch_G \comp \inj_x, \ch_G \comp \inj_x, \turno_G \comp \inj_x, \vince_G \comp \inj_x, \inj_x^{-1}(G_\Player),  \inj_x^{-1}(G_\Opponent))$ is the sub-game of $G$ of root $x$. 
\item
If $m$ is any initial move of $G$, the immediate sub-game of $G$ defined by $m$ is $G_{\langle m \rangle}$.
\item
The dual game of $G$ is $G^\bot = (|G|, \Ch_G, \ch_G,  \turno^\bot_G, \vince^\bot_G, G_{\Opponent}, G_{\Player})$.
\end{enumerate}
\end{definition}

For all $G \in \setP$, by $\FV(G)$ self-dual we have $\FV(G) = \FV(G^\bot) = \FV(G)^\bot$. 
By definition, for all $G \in \setP$ we have $G_{\nil} = G$ and $G^{\bot\bot} = G$ and $G \not = G^\bot$: the map $(.)^\bot$ on $\setP$ is involutory and without fixed points. We have $G \in \setG$ ($G$ is set-theoretical, without parameters) if and only if $G^\bot \in \setG$.

Assume $p = \langle x_0, x_1, x_2, \ldots \rangle \in |G|_{\le \infty}$ is any finite or infinite play of $G$. $p$ starts from the root of $G$, the position number $i$ of $p$ is $x = \langle x_0, x_1, \ldots, x_{i-1}\rangle$. If $x$ is not a leaf then then the player $g = \turno_G(x)$ selects the next move $x_i$.

A terminated play is any maximal list $p$ of $|G|_{\le \infty}$. If $p \in |G|$ is finite then $p$ is a leaf of $|G|$ and the winner is  $\vince_G(p)$, if $\vince_G(p) \in \{\Player, \Opponent\}$, otherwise the winner is decided w.r.t. some $G$-assignment $\rho$. If $p \in |G|_{\infty}$ is infinite, then the winner is $\Player$ if $p \in G_{\Player}$, is $\Opponent$ if $p \in G_{\Opponent}$, independently from $\rho$. 



\subsection{A notion of game strategy}
A strategy $\sigma$ is a particular set of plays. When $g$ is using a strategy $\sigma$, and $p \in \sigma$ and $\turno_G(p) = g$, then we think of the one-step extensions $q = p@m$, $q \in \sigma$ of $p$ in $\sigma$ as the suggestions of $\sigma$ for a move $m$ of $g$ from $p$. The suggestions may consist of no move, of one move, of two or more possible moves. If $\turno_G(p) = g^\bot$, then we think of the one-step extensions $q = p@m$, $q \in \sigma$ as the set of replies of $g^\bot$ considered by $\sigma$. These replies may not be all possible replies of $g^\bot$. Formally, a strategy for $G$ is any tree support included in  $|G|$.

\begin{definition}[Strategies]
Assume $G \in \setP$ be any set-theoretical game. 
\begin{enumerate}
\item $\sigma$ is a $G$-strategy if $\sigma$ is a tree support and $\sigma \subseteq |G|$. We write $\sigma:G$ for ``$\sigma$ is a $G$-strategy''. 
\item
$\sigma$ follows $\tau$ on $x \in \sigma$ if $\tau = \sigma_x$.
\end{enumerate}
\end{definition}

We informally describe some desirable features of strategies. $\sigma$ is a $g$-strategy if $\sigma$ takes in to account all moves of the opponent of $g$. A $g$-strategy $\sigma$ is $g$-total if $\sigma$ always suggests some move when $g$ should move in $\sigma$. A $g$-total $\sigma$ is $g$-partially winning if $\sigma$ wins all finite maximal plays in $\sigma$. A $g$-partially winning $\sigma$ is $g$-winning if $g$ wins all infinite maximal plays.  

\begin{definition}[Winning strategies]
Assume $G \in \setP$ and $\sigma:G$ is a strategy on $G$. Let $g \in \{\Player, \Opponent\}$ be a player.
\begin{enumerate}
\item
$\sigma$ is a $g$-strategy if for all $p \in \sigma$, if $\turno_G(p) = g^\bot$ then for all $(q >_1 p)$, $(q \in |G|)$ we have $q \in \sigma$. 
\item 
$\sigma$ is $g$-total if $\sigma$ is a $g$-strategy and for all $p \in \sigma \setminus \Foglie(|G|)$ if $\turno_G(p) = g$ then there is some $q >_1p$, $q \in \sigma$. 
\item
$\sigma$ is $g$-partially winning if $\sigma$ is $g$-total and for all $p \in \sigma \cap \Foglie(|G|)$ we have $\vince_G(p) = g$. 
\item
$\sigma$ is $g$-winning if and only if $\sigma$ is $g$-partially winning and $\sigma_{\infty} \subseteq G_g$ (all infinite branches in $\sigma_{\infty}$ are won by $g$). 
\end{enumerate}
\end{definition}

%

By definition, $\sigma$ $g$-winning implies that $\sigma$ is $g$-total. Let $G \in \setP$ be any game. We say that $G$ is $g$-winning if there is some $g$-winning strategy $\sigma$. $G$ is determined if $G$ is $g$-winning for some $g \in \{\Player, \Opponent\}$. Games in $\setP$ may not be determined because the winner of some leaves is not settled. If we assume the Choice axiom, then there are games in $\setG$ (set-theoretical, without parameters) which are not determined. Remark that $G$ is $g$-winning if and only if $G^\bot$ is $g^\bot$-winning.

 We denote with $\vince(G)$ the winner of $G$ if $G$ is determined, otherwise we let $\vince(G)$ undefined.


 If $\sigma:G$ (i.e., if $\sigma \subseteq |G|$) and $x \in \sigma$, then $\sigma_x$ is a sub-tree support of $\sigma$ and $\sigma_x \subseteq |G_x|$, therefore $\sigma_x : G_x$. $\sigma_x$ is a strategy for the sub-game of $G$ of root $x$. By definition unfolding, if $\sigma:G$ is a $g$-strategy, is $g$-total, $g$-partially winning, $g$-winning for $G$, then $\sigma_x : G_x$ is, respectively: a $g$-strategy, is $g$-total, $g$-partially winning, $g$-winning for $G_x$. To check that $\sigma_x$ is $G_x$-winning, we use the fact that, by definition, we have $(G_x)_g = \inj_x^{-1}(G_g)$. That, is, the infinite plays of $G_x$ which are $g$-winning are exactly the counter-images of the infinite plays of $G$ which are $g$-winning.

\subsection{Tarski games: a game interpretation for first order connectives}
We define some operations on the set $\setP$ of parametric games corresponding to truth values, boolean connectives, first order quantifiers in logic. Using these operations we may interpret any first order closed arithmetical formula $A$ by some set-theoretical game $G$, in such a way that $G$ is $\Player$-winning if and only if $A$ is true. This game semantics is called Tarski games 
it is sound and complete for first order arithmetic, but it lacks an interpretation for second order quantifiers $\FORALL X.A[X]$, and lacks primitive recursive $\Player$-winning strategies for many true formulas $A$.\footnote{Indeed, assume that $p(x,y,z)$ is the primitive recursive predicate stating that $f_x$, the partial recursive map number $x$, when applied to $y$ terminates in $z$ steps. Then the formula $A = \forall x,y.(\exists z.p(x,y,z)) \vee (\forall z.p^\bot(x,y,z))$ states that either $f_x(y)$ terminates or not. The Tarski game for $A$ has $\Player$-winning strategies, but all of them decide the Halting Problem and therefore are not recursive.} 

Assume we have some at most countable family of games $\{G_i | i \in I\}$ interpreting a family of formulas $\{A_i | i \in I\}$. We will define a game $G = \vee_{i \in I} G_i$ interpreting the truth of the possible infinite disjunction $A = \vee_{i \in I} A_i$. $\Player$ plays first in $G$, selecting some $i \in I$, then the plays goes on as in $G_i$.
If $I = \emptyset$ then $\Player$ cannot move in $G$ and $\Opponent$ wins: $G$ interprets the constant $\False$. If $I$ has two elements then $G$ interprets a binary disjunction and if $I$ is countable then $G$ interpret an existential over $\Nat$. $G = \wedge_{i \in I} G_i$ is the dual game, interpreting some possibly infinite conjunction, interpreting the constant $\True$ if $I = \emptyset$, the binary conjunction if $I$ has two elements, the universal quantifier on $\Nat$ if $I$ is countable.

\begin{definition}[Conjunctions and disjunctions of parametric games]
\mbox{} If $I \subseteq \Nat$ and $\{G_i | i \in I\} \subseteq \setP$, then $G = \vee_{i \in I} G_i \in \setP$ is defined as follows. 
\begin{enumerate}
\item
$|G| = \{\nil\} \cup \{i@x | (i \in I) \wedge (x \in |G_i))|\}$ (\emph{$|G|$ is the tree whose immediate subtrees are all $|G_i|$})
\item
If $I = \emptyset$ then $\vince_G(\nil) = \Opponent$ . If $I \not = \emptyset$ then $\turno_G(\nil) = \Player$ (\emph{$\Player$ moves first})
\item
for all $i \in I$, all $x \in |G_i|\setminus \Foglie(|G_i|)$: $\turno_G(i@x) = \turno_{G_i}(x)$ (\emph{the game continues in some $G_i$})
\item
for all $i \in I$, all $x \in \Foglie(|G_i|)$: $\vince_G(i@x) = \vince_{G_i}(x)$ (\emph{winning conditions are taken from each $G_i$})
\item
$G_g = \{i @ x | (i \in I) \wedge (x \in (G_i)_g) \}$ for all $g \in \{\Player, \Opponent\}$ (\emph{winning conditions are taken from each $G_i$})
\end{enumerate}
We set $\wedge_{i \in I} G_i = (\vee_{i \in I} G^\bot_i)^\bot$.
\end{definition}

%
%

Let $G \in \setP$. We say that $G$ is an atomic game, or just \emph{atomic}, if and only if $|G| = \{\nil\}$ is an atomic tree (is the one-node tree). The unique node of $G$ is a leaf, all plays have $0$ moves and in order to precise $G$ we only have to precise the winner (the value of $\vince_G(\nil)$). Let $x \in \{\Player, \Opponent\} \cup \Par$: we define $G = \END(x) $ as the atomic game such that $\vince_G(\nil) = x$. We have $\END(\Player) = \wedge_{i \in \emptyset} G_i$ and $\END(\Opponent) = \vee_{i \in \emptyset} G_i$ and $\END(x)^\bot = \END(x^\bot)$. We call $\END(a)$ for $a \in \Par$ a \emph{generic game}. The winner of $\END(a)$ is given by the value $\rho(a) \in \{\Player, \Opponent\}$ that some ``judge of the play'' $\rho$ assigns to the parameter $a$. $G$ is a atomic if and only if $G = \END(x)$ for some $x \in \{\Player, \Opponent\} \cup \Par $.

If $c = \vee, \wedge$, then the root of $c_{i \in I} G_i$, if it is a leaf, it is labeled by $\Player$ or $\Opponent$, not by a parameter. We deduce that $\FV(c_{i \in I} G_i) = \cup_{i \in I} \FV(G_i)$. In particular, if all $G_i$ are in $\setG$ (if $\FV(G_i) = \emptyset$ for all $i \in I$) then $\FV(c_{i \in I} G_i) = \cup_{i \in I} \emptyset = \emptyset$, hence $c_{i \in I} G_i \in \setG$. We proved that the set $\setG \subsetneqq \setP$ of set-theoretical games is closed under the operations: $\END(g)$ for $g = \Player, \Opponent$ and $\vee (.)$, $\wedge (.)$.

For any parametric non-generic game $G$ (for all $G \in \setP$ such that $G \not = \END(a)$ for all $a \in \Par$) we have the following characterization. $G$ is $c_{i \in I} G_{\langle i \rangle}$, for some $c = \vee, \wedge$, where: $I$ the set of initial moves of $G$, and $G_{\langle i \rangle}$ is the immediate sub-game of $G$ defined by the initial move $i$.


\subsection{Discussing a game interpretation of second order quantification}
In this sub-section we informally outline our game interpretation for a second order quantification $\FORALL X.A$: in the next sub-section we will make it precise. 


Assume that we interpreted the predicate variable $X$ with some self-dual set of parameters $I$, and that we have a game $\interpret{A}$ interpreting $A$. We define a game $G$ interpreting $\FORALL X.A$ as follows. $\Player$ wins a play in the game $G$ if: 
\begin{enumerate}
\item
either $\Player$ wins $\interpret{A}$ independently from the assignment to the parameters in $I$, or
\item 
for some $a \in I$, $\Player$ is able to find two dual generic sub-games $\END(a)$ and $\END(a^\bot)$ in $\interpret{A}$, such that if $\Player$ wins $\END(a)$ then player wins $\interpret{A}$, and if $\Player$ wins $\END(a^\bot)$ then player wins $\interpret{A}$.
\end{enumerate}

In the second case, $\Player$ proved that for any assignment to the parameters in $I$ there exists a $\Player$-winning strategy for $\interpret{A}$, even if this $\Player$-winning strategy depends on the assignment and is not known by $\Player$. 

Now we describe $G$ more in detail. Let $\vec{X} = X_1, \ldots, X_m$ be any list of predicate variables. We interpret the connectives $\FORALL X.A$ and $\EXISTS X.A$ of $L(\PA^2)$ as a particular case of the more general connectives $\FORALL \vec{X}.\Gamma$ and $\EXISTS \vec{X}.\Gamma$, with $\Gamma = (A_0, \ldots, A_{n-1})$ a sequent. $\FORALL \vec{X}.\Gamma$ and $\exists \vec{X}.\Gamma$ are more suitable to an effective interpretation, and their meaning is, respectively, $\FORALL \vec{X}.A_0 \vee \ldots \vee A_{n-1}$ and $\exists \vec{X}. A_0 \wedge \ldots \wedge A_{n-1}$. In game theoretical terms, we debate the truth of $\FORALL \vec{X}.(A_0, \ldots, A_n)$ by interleaving several ``local'' debates about the truth of $A_0, \ldots, A_{n-1}$, and considering ``unknown'' the truth value of each instance $X_i(\vec{c_n})$ of each $X_i$. Our goal is convincing our opponent that, no matter how we assign a list of predicates to $\vec{X}$, some $A_i$ is true. Sometimes we obtain this effect by finding two dual generic sub-games $\END(a)$ and $\END(a^\bot)$ in $A_0, \ldots, A_{n-1}$, with $a$ in the interpretation of $\vec{X}$, sometimes by finding some atomic formula which is true independently from the predicate assignment to $\vec{X}$. At each even step $\Player$ selects some $A_i$ and in the next step (which is an odd step) the debate continues from $A_i$. The move from $A_i$ to some $A'_i$ creates a new local debate about $A'_i$, but does not delete $A_0, \ldots, A_{n-1}$: at any moment, $\Player$ may move again from $A_i$. The fact that we may come back to $A_i$ is called ``backtracking'' by Coquand \cite{Coquand-1991}, \cite{Coquand-1995}.

We give a move precise description of how to define $G$. We call any $\Gamma = (G_0, \ldots, G_{n-1}) \in \setP^n$ a \emph{game sequent}, and we set $\Gamma^\bot = (G^\bot_0, \ldots, G^\bot_{n-1})$ and $\FV(\Gamma) = \FV(G_0) \cup \ldots \cup \FV(G_{n-1})$. Assume some self-dual set $I \subseteq \Par$ of parameters be given. $I$ interprets a list of predicate variables in $L(\PA^2)$: we call $I$ the set of \emph{bound parameters} of $G$. We want to define some game $G = \FORALLGAME I.\Gamma \in \setP$ with free parameters $\FV(G) = \FV(\Gamma) \setminus I$, interpreting universal quantification on predicates. Remark that $\FV(\Gamma) \setminus I \subseteq \Par$ is a self-dual set, because $\FV(\Gamma)$ and $I$ are self-dual. 

%

The game $G$ runs as follows. In any moves with even index $2a = 0, 2, 4, \ldots$, player $\Player$ moves, and selects some non-atomic game $G_i$. For a reason we explain in a moment, we name this move $\JUST(i,n)$. The player moving in the move of odd index $2a+1$ is defined as the first player of $G_i$. If $\Player$ is the first player in $G_i$, $\Player$ moves $m$ from the root of $G_i$, selecting the immediate sub-game $G_n = (G_{i})_{\langle m \rangle}$. If $\Opponent$ is the first player in $G_i$, then $\Opponent$ moves $m$ from the root of $G_i$, selecting the immediate sub-game $G_n = (G_{i})_{\langle m \rangle}$. In both cases the sub-game of $G$ we obtain is equal to $G' = \FORALLGAME I.(G_0, \ldots, G_{n-1}, G_n)$. We call each $G_i$ a local position of $G$,
and we say that the local position number $i$ \emph{justifies} the existence of the local position number $n$. For this reason the move of index $2a$ by $\Player$, selecting the sub-formula $G_n$ of $G_i$, is called $\JUST(i,n)$. We call a ``local play'' the part of the play on $G$ running on some local position $G_i$.

There is the special move $\DROP(n)$ for $\Player$, with $n$ the length of the list $(G_0, \ldots, G_{n-1})$: if $\Player$ moves $\DROP(n)$ then $\Player$ ``gives up''. The goal of $\Player$ is to reach some sub-game $\FORALLGAME I.(G_0, \ldots, G_{m-1})$ in which, \emph{for any assignment to the parameters in $I$}, player $\Player$ wins some atomic game $G_i$. In a finite play, this goal is achieved by two more special moves, $\EM(i,j)$ or $\STOP(i)$. 

\begin{enumerate}
\item[$\EM(i,j)$]
$\Player$ chooses some $i, j \in \Nat$, $i, j < m$ in some sub-game such that $G_i = \END(a)$ and $G_j = \END(a^\bot)$ for some $a \in I$.
No matter how we assign the parameters in $I$,
$\Player$ wins either in $G_i$ or in $G_j$. If $\Player$ finds such a configuration, then we say that $\Player$ wins $G$. We call this move $\EM(i,j)$.
\item[$\STOP(i)$]
$\Player$ chooses, in some sub-game, some $i \in \Nat$ with $i < n$, such that $G_i = \END(g)$ or $G_i = \END(a)$, for some $g \in \{\Player, \Opponent\}$ and some $a \in \FV(G)$. We call this move $\STOP(i)$: the game stops in $G_i$, the winner is decided by the label of the root of $G_i$ and, possibly, by the assignment $\rho$ on $\FV(G)$. 
\end{enumerate}

In the definition of $\STOP(i)$, we asked that $a \in \FV(G) = \FV(\Gamma) \setminus I$, hence that $a \not \in I$. The reason is that, in the case the play ends $G_i = \END(a)$ with a move $\STOP(i)$, the value assigned to the parameter $a$ decides the winner, and we do not want the winner of $G$ to depend on a bound parameter of $G$. The meaning of a quantifier does not depend on an assignment to its bound variable: you cannot assign a bound variable. We want that the same holds for the bound parameters of a game interpreting a quantification. 


In infinite plays, $\Player$ wins a play $p$ on $G = \FORALLGAME I. (G_0, \ldots, G_{n-1})$ if and only if, for some $i<n$, $\Player$ wins some sub-play $q$ of $p$ made by all moves of $p$ which are in $G_i$. 


\subsection{A game interpretation of second order quantification}
In this sub-section we formally define a set $\setM$ of moves and a game $\FORALLGAME I.\Gamma$ on $\setM$, interpreting second order quantification.


The constructors of $\setM$ are $\DROP(n), \EM(n,m), \STOP(n), \JUST(n,m), \nth(n)$, with arguments any $n, m \in \Nat$. The constructors are all unary or binary, and we read them as follows. $\DROP(n)$: ``I drop the discussion of $n$ formulas''. $\EM(i,j)$: ``I use Excluded Middle on atomic formulas number $i$, $j$''. $\STOP(i)$: ``the atomic formula number $i$ gives the outcome of the play''. $\JUST(i,n)$: ``The non-atomic formula number $i$ justifies the formula number $n$''. For interpreting formulas with a propositional or first-order head symbol we add the unary constructor $\nth(n)$: ``I choose the immediate sub-formula number $n$''. 

We choose $\setM$ in such a way that the constructors are primitive recursive, have disjoint range, a primitive recursive inverse for each argument, and each move in $\setM$ is in the image of some constructor. This is possible. 


Let $G = \FORALLGAME I.(G_0, \ldots, G_{n-1})$. We translate the discussion of the previous section into an inductive definition of the set of $p \in |G|$, of the list of local positions $G_0, \ldots, G_{(n-1)+k}$ of $p$, extending $G_0, \ldots, G_{n-1}$, and of the justification relation between indexes of local positions.

\begin{definition}[The tree support of $G = \FORALLGAME I.\Gamma$] 
\mbox{}
\begin{enumerate}
\item
$\nil \in |G|$ and $\nil$ has local positions $G_0, \ldots, G_{n-1}$. 
\item
Assume that: $p = \langle \JUST(i_0,n), m_0, \ldots, \JUST(i_{k-1},n+k-1), m_{k-1} \rangle$ has length $2k$ and $p \in |G|$ and $p$ has local positions $G_0, \ldots, G_{n+k-1}$ and $m \in \setM$. Then $p_1 = p@m \in |G|$ if and only if one of these conditions holds:
\begin{enumerate}
\item
$m = \DROP(n+k)$.
\item
$m = \EM(i,j) \in |G|$ and $i, j < n+k$ and $G_i = \END(a)$, $G_j= \END(a^\bot)$ for some $a \in I$.
\item
$m = \STOP(i)$ and $i < n+k$ and either $G_i = \END(g)$ or $G_i = \END(a)$, for some $g \in \{\Player, \Opponent\}$ and some $a \in \FV(G)$ (hence $a \not \in I$).
\item
$m = \JUST(i_k,n+k)$ and $i_k < n+k$ and $G_{i_k}$ is \underline{not} atomic.
\end{enumerate}
The local positions of $p_1$ are those of $p$. 
\item
If $p_1 = p@\JUST(i_k,n+k) \in |G|$ then $p_2 = p_1@m_k \in |G|$ if and only if $m_k$ is a move from $G_{i_k}$. The local positions of $p_2$ are $G_0, \ldots, G_{n+k}$, with $G_{n+k} = (G_{i_k})_{\langle m_k \rangle}$. 
\end{enumerate}
For all $j \le k$, we say that \emph{$i_j$ justifies $n+j$ through $m_j$} and we write $ i_j \vdash_{m_j} n+j $. 
\end{definition}

We define now the maps $\turno_G$ and $\vince_G$ for $G = \FORALLGAME I.\Gamma$. 

\begin{definition}[Turn and winner for nodes of $G$]
Assume $p$ has even length: then $\turno_G(p) = \Player$. Take any $p @ \DROP(n+k), p @ \EM(i,j), p@ \STOP(i)$ and $p @ \JUST(i,n+k)$ of odd length in $|G|$. Then we set:
\begin{enumerate}
\item
$\vince_G(p @ \DROP(n+k)) = \Opponent$
\item
$\vince_G(p@ \EM(i,j)) = \Player$
\item
$\vince_G(p @ \STOP(i)) = \vince_{G_i}(\nil)$. 
\item
$\turno_{G}(p @ \JUST(i,n+k)) = \turno_{G_i}(\nil)$ (\emph{on odd positions, the player moving is the first player of the local position $G_i$})
\end{enumerate}
We call any $a \in I$ a bound parameter of $G$.
\end{definition}




In order to complete the definition of $G = \FORALLGAME I.\Gamma$ we have to define the set $G_g$ of infinite plays of $G$ won by $g \in \{\Player, \Opponent\}$.
We want to define $G_g$ in such a way that $\Player$ wins an infinite play in $G$ if and only if $p$ wins some local play in some $G_0, \ldots, G_{n-1}$. To this aim, we need a ``local play'' relation $q \prec_i p$ relating a play $q \in |G_i|_{\le \infty}$ with a play $p \in |G|_{\le \infty}$. 

We first extend the justification relation by reflexivity and transitivity, as follows. Assume $i < n$, $q = \langle m_0, \ldots, m_{n-1} \rangle \in |G_i|$, $j \in \Nat$. We define a relation $i \vdash_q j$ by: there are $i = i_0 < \ldots < i_n = j$ such that $i_0 \vdash_{m_0} i_1$, \ldots, $i_{n-1} \vdash_{m_{n-1}} i_{n}$. We say that $i$ justifies $j$ if $i \vdash_m j$ for some for some move $m$ from $G_i$. We say that $i$ remotely justifies $j$ if $i \vdash_q j$ for some $q \in |G_i|$. In this case we also say that $q$ is a local play of $p$ in $G_i$ and we write $q \prec_i p$. If $q = \langle m_0, \ldots, m_{n-1}, \ldots \rangle \in |G_i|_\infty$ we say that $q \prec_i p$ if there are infinitely many $i = i_0 < \ldots < i_n < \ldots$ such that $i = i_0 \vdash_{m_0} i_1$, \ldots, $i_{n-1} \vdash_{m_{n-1}} i_{n}$, \ldots. From the local play relation $\prec_i$ for $i<n$ we may define $G_{\Player}$, $G_{\Opponent}$. 


\begin{definition}[The partition $G_{\Player}$, $G_{\Opponent}$]
Let $G = \FORALLGAME I.(G_0, \ldots, G_{n-1})$.
\begin{enumerate}
\item
$G_{\Player}$ is the set of infinite plays $p \in |G|_{\infty}$ such that for some $i \in \Nat$, $i<n$, some infinite $q \prec_i p$ we have $q \in (G_i)_{\Player}$.
\item
$G_{\Opponent}$ is the set of infinite plays $p \in |G|_{\infty}$ such that for all $i \in \Nat$, $i<n$, all infinite $q \prec_i p$ we have $q \in (G_i)_{\Opponent}$.
\end{enumerate}
\end{definition}

This ends the definition of $\FORALLGAME I.\Gamma$. We define $\EXISTSGAME I.\Gamma$ as the \emph{dual game $(\FORALLGAME I.\Gamma^\bot)^\bot$}.
\\

The operator on games $\FORALLGAME I.\Gamma$ is quite different from the game operators we have for Tarski games. Assume that $G_0, \ldots, G_{n-1} \in \setG$. Then $\FV(G) = \emptyset$, therefore both $G$ and $G_0 \vee \ldots \vee G_{n-1}$ are in $\setG$. If $G_0, \ldots, G_{n-1}$ are determined, we will prove that both games have a $\Player$-winning strategy if and only if there is a $\Player$-winning strategy for some $G_i$. Thus, if $G_0, \ldots, G_{n-1} \in \setG$, then both $G$ and $G_0 \vee \ldots \vee G_{n-1}$ are an interpretation of the logical disjunction. Even in this case, $G$ and $G_0 \vee \ldots \vee G_{n-1}$ have very different features. Under the assumption $\FV(G) = \emptyset$, we will prove that $G$ is determined, and that $G$ has a primitive recursive $\Player$-winning strategy whenever it has a $\Player$-winning strategy and $G_0, \ldots, G_{n-1}$ have primitive recursive finitary part. Both properties fail for $G_0 \vee \ldots \vee G_{n-1}$, even if we assume that $G_0, \ldots, G_{n-1} \in \setG$, even if $n=1$.

\section{The language of $\PA^2$ and its notion of truth}
\label{section-pa}
In this section we define a language $L(\PA^2)$ for classical second order arithmetic $\PA^2$, with terms, formulas, predicates, one-sided sequents, and substitutions. Negations are pushed to atomic formulas and implications are defined from negations. 
Eventually, we define the canonical notion of validity and truth for formulas of $L(\PA^2)$, and our game semantics for $L(\PA^2)$. 

 \subsection{First order terms}
 We consider a set $\Term$ of first order terms, defined as follows. Let $\Nat$ denote the set of natural numbers. For every integer $k \in \Nat$, for every primitive recursive map $\phi:\Nat^k \rightarrow \Nat$ we assume having a function symbol $f$ denoting it. When $k = 0$, for every natural number $n \in \Nat$ we assume having a constant $c_n$ denoting it. We have infinitely many variables $x_0, x_1, x_2, \ldots$ denoting elements of $\Nat$. 


\subsection{Second order formulas, sequents and predicates}
 The set of terms of $L(\PA^2)$ is $\Term$. Let $\Bool = \{\true, \false\}$ be the set of booleans values. We define the dual operation $\true^\bot = \false$ and $\false^\bot = \true$. $(.)^\bot$ is involutory and without fixed point: we have $b^{\bot\bot} = b$ and $b \not = b^\bot$ for all $b \in \Bool$. If $f:I \rightarrow \Bool$ is any map, we define the dual $f^\bot$ of $f$ point-wise, by $f^\bot(i) = f(i)^\bot$ for all $i \in I$. If $f$ denotes a predicate on $I$, then $f^\bot$ denotes the complement (the negation) of the predicate. For all $k \in \Nat$, any primitive recursive $k$-ary predicate $\psi :\Nat^k \rightarrow \Bool$ we have in $L(\PA^2)$ two distinct symbols, $p$ (positive) and $p^{\bot}$ (negative), denoting $\psi$ and the dual predicate $\psi^\bot$ (the complement or negation of $p$). When $k=0$ there is a constant for truth, which we denote with $\True$. We write $\False$ for the dual constant $\True^{\bot}$. For all $k \in \Nat$ the language $L(\PA^2)$ has infinitely many $k$-ary predicate variables $V^k_0, V^k_1, \ldots, V^k_i, \ldots$. We usually drop the superscript $k$ and the index $i$, and we write $X, Y, Z, \ldots$ to denote a predicate variable. For every variable $X \in L(\PA^2)$, denoting a predicate $\psi : \Nat^k \rightarrow \Bool$, we add to $L(\PA^2)$ the negated variable $X^{\bot}$ denoting the complement $\psi^\bot$ of $\psi$. An atomic formula $A$ is $p(\vec{t})$ or $p^{\bot}(\vec{t})$ or $X(\vec{t})$ or $X^{\bot}(\vec{t})$, where $p$, $X$ are a constant and a predicate variable of arity $k \in \Nat$ and $\vec{t} = t_1, \ldots, t_k$ is a list of $k$ terms in $\Term$. 

We consider a negation- and implication-free language, where all negation are pushed to the atomic formulas, and represented by adding/removing superscript $(.)^\bot$ on the predicate symbols.

%
%

\begin{definition}[Formulas of $L(\PA^2)$]
Assume $A,B \in L(\PA^2)$ are formulas, $x$ is any variable over $\Nat$, and $X$ is any predicate variable. 
\begin{enumerate}
\item
Any atomic formula is some formula of $L(\PA^2)$. 
\item
$A \wedge B, A \vee B, \forall x.A, \exists x.A, \FORALL X.A, \EXISTS X.A \in L(\PA^2)$ are formulas. 
\end{enumerate}
\end{definition}


A sequent of $L(\PA^2)$ is any list $\Gamma = (A_0, \ldots, A_{n-1})$ of formulas of $L(\PA^2)$. We consider the one-formula sequent $(A)$ distinct from $A$ itself.
We call any formula $A_1 \vee A_2, \exists x.A, \EXISTS X.A \in L(\PA^2)$ a disjunctive formula of $L(\PA^2)$. 
We call any formula $A_1 \wedge A_2, \forall x.A, \FORALL X.A \in L(\PA^2)$ a conjunctive formula of $L(\PA^2)$. 

\begin{definition}[Head and Order]
\label{definition-head}
Let $A \in L(\PA^2)$ be any formula.
\begin{enumerate}
\item
The head of $A$ is the outermost symbol of $A$
\item
$A$ has order $0$ if $A$ has head some predicate variable $X$, $X^\bot$, or some predicate constant $p$, $p^\bot$, or some connective $\vee, \wedge$
\item
$A$ has order $1$ if $A$ has head some connective $\exists, \forall$
\item 
$A$ has order $2$ if $A$ has head some connective $\EXISTS, \FORALL$. 
\end{enumerate}
\end{definition}

 We define free and bound variables of $L(\PA^2)$ as usual. We denote by $\FV(A)$ the set of variables $x, X, X^\bot$ occurring free in $A$, or such that the dual of the variable occurs free in $A$. We denote with $\FV_1(A), \FV_2(A)$ the subsets of first order and second order variables in $\FV(A)$.  A formula $A \in L(\PA^2)$ is $1$-closed if $\FV_1(A) = \emptyset$. We define the substitution $[\vec{t}/\vec{x}]$ on integer variables as usual, using variable renaming in order to avoid variable capture. For any $k \in \Nat$, a $k$-ary predicate is any expression $\lambda \vec{x}.A$, for some list $x_1, \ldots, x_k$ of variables over $\Nat$ and some formula $A$. When $k = 0$, the $0$-ary predicates are exactly the formulas.

 \subsection{Negation, Implication and substitution}
There is no primitive negation over $L(\PA^2)$, but we define an involutory negation $A^\bot$ as follows. $A^\bot$ is obtained by switching in $A$: $X$ with $X^\bot$, $p$ with $p^\bot$, $\wedge$ with $\vee$, $\forall$ with $\exists$, $\FORALL$ with $\EXISTS$. As a consequence, $A^{\bot\bot} = A$ and $(.)^\bot$ has no fixed point: $A \not = A^\bot$ for all $A$. $A$ is atomic if and only if $A^\bot$ is atomic, and $A$ is conjunctive (disjunctive) if and only if $A^\bot$ is disjunctive (conjunctive). We define $(A \Rightarrow B) = (A^{\bot} \vee B)$.


Assume $\vec{P} = P_1, \ldots, P_n$ is a list of predicates, with $P_i = \lambda \vec{x_i}.A_i$ of arity $k_i$, and $\vec{X} = X_1, \ldots, X_n$ is a list of predicate variables of the same arity. We define a substitution $[\vec{P}/\vec{X}]$ as follows. We set $X_i[\vec{P}/\vec{X}] = A_i[\vec{t}/\vec{x}]$ and $X^\bot_i[\vec{P}/\vec{X}] = A^\bot_i[\vec{t}/\vec{x}]$. If $A$ is atomic and $\FV(A) \cap \vec{X} = \emptyset$ then we set $A[\vec{P}/\vec{X}] = A$. We extend $[\vec{P}/\vec{X}]$ by compatibility with formula construction in $L(\PA^2)$, using variable renaming in order to avoid variable capture.

Assume we have involutory operations denoted $(.)^\bot$ on $I,J$. We say that a map $f:I \rightarrow J$ is \emph{dual-preserving} when $f(i^\bot) = f(i)^\bot$ for all $i \in I$. Substitution is an example of dual-preserving map on formulas $A \in L(\PA^2)$: by induction on $A$ we may prove that $A^\bot[\vec{P}/\vec{X}] =  A[\vec{P}/\vec{X}]^\bot$.

\subsection{Interpreting formulas of $L(\PA^2)$}
We formally describe the standard model-theoretical interpretation for a formula $A \in L(\PA^2)$, in which a first order variable $x$ has domain $\Nat$ and a $k$-ary second order predicate variable $X^k_n$ has domain \emph{all} subsets of $\Nat^k$. We assume having a dual-preserving map $[.]$, taking any $k$-ary function symbol $f \in L(\PA^2)$, any $k$-ary predicate positive symbol $p \in L(\PA^2)$, and returning the primitive recursive map $[f]:\Nat^k \rightarrow \Nat$ which is denoted by $f$, and the primitive recursive predicate $[p]:\Nat^k \rightarrow \Bool$ which is denoted by $p$. In the special case $k=0$, then $f = c_n$ is a constant denoting the natural number $[c_n] = n \in \Nat$, and $p$ is a constant predicate, denoting some boolean $[p] \in \Bool$. We assumed that $[.]$ is dual-preserving: this means that 
$[p^\bot] = [p]^\bot$.

An $A$-environment $\theta$ is any dual-preserving map $\theta:\dom(\theta) \rightarrow \Nat \bigcup (\cup_{k \in \Nat} (\Nat^k \rightarrow \Bool))$, such that $\FV(A) \subseteq \dom(\theta)$ and $\theta(x) \in \Nat$ for all $x \in \FV_1(A)$ and $\theta(X) : \Nat^k \rightarrow \Bool$ for all $X \in \FV_2(A)$ with arity $k$. If $A$ is closed then the empty map $\emptyset$ is an $A$-environment. We write $\setE(A)$ for the set of $A$-environments and $\setE$ for the set of environments. If $\rho_1$, $\rho_2$ are two environments, we define the over-writing $\rho_1, \rho_2$ of $\rho_1$ with $\rho_2$ by $(\rho_1,\rho_2)(x) = \rho_2(x)$ if $x \in \dom(\rho_2)$, and $(\rho_1,\rho_2)(x) = \rho_1(x)$ if $x \in \dom(\rho_1) \setminus \dom(\rho_2)$. 

Given any $t = x, f(\vec{t}) \in \Term$, with $\vec{t} = t_1, \ldots, t_n$, any environment $\theta$ with $\FV(t) \subseteq \dom(\theta)$, we recursively define $[t]_{\theta} \in \Nat$ and $[\vec{t}]_{\theta} \in \Nat^n$ by: $[x]_{\theta} = \theta(x)\in \Nat$ and  
$[f(\vec{t})]_{\theta} = 
[f]( [\vec{t}]_{\theta} )\in \Nat$ and 
$[\vec{t}]_{\theta} =  [t_1]_{\theta}, \ldots, [t_n]_{\theta}\in \Nat^n$.
We may now define the dual-preserving interpretation $[A]_{\theta} \in \Bool$ of $A$, given any $\theta \in \setE(A)$.

\begin{definition}[Interpretation of a formula of $L(\PA^2)$]
\mbox{} Let $X$ be any predicate variable, and $x$ be any variable on $\Nat$. Assume $A \in L(\PA^2)$ is any formula and $\theta  \in \setE(A)$ is any environment for $A$. 
We define $[A]_{\theta} \in \Bool$ by induction on $A$. 
\begin{enumerate}
\item
$[s(\vec{t})]_{\theta} 
= [s]([\vec{t}]_{\theta}) \in \Bool$ if $s = p, p^\bot$ (recall that $[.]$ is dual-preserving)
\item
$[s(\vec{t})]_{\theta} 
= \theta(s)([\vec{t}]_{\theta}) \in \Bool$
if $s = X_i, X_i^\bot$ (recall that $\theta$ is dual-preserving)
\item
$[A_1 \wedge A_2]_{\theta}  = \true$ 
if and only if $[A_i]_{\theta} = \true$ for all $i \in \{1,2\}$. 
\item
$[\forall x.A]_{\theta} = \true$ 
if and only if $[A[c_n/x]]_{\theta} = \true$ for all $n \in \Nat$
\item
$ [\FORALL X.A]_{\theta} = \true$ if and only if for all $\eta \in \setE$ with $\dom(\eta) = \{X\}$ we have $[A]_{\rho,\eta} = \true$. 
\item
If $A$ is disjunctive then we set: $[A]_{\theta} = [A^\bot]^\bot_{\theta}$. 
\end{enumerate}
\end{definition}

We check that $[.]_\theta$ is a dual-preserving map, from any $A \in L(\PA^2)$ such that $\FV(A) \subseteq \dom(\theta)$ into $\Bool$. Indeed, if $A$ is atomic then $[A]^\bot_{\theta} = [A^\bot]_{\theta}$ because $[.]$ and $\theta$ are dual-preserving. If $A$ is disjunctive, then $[A]_{\theta} = [A^\bot]^\bot_{\theta}$ by definition, therefore $[A]^\bot_{\theta} = [A^\bot]^{\bot\bot}_{\theta} = [A^\bot]_{\theta}$. If $A$ is conjunctive, then $A^\bot$ is disjunctive with dual $A^{\bot\bot} = A$, and by definition we have $[A^\bot]_{\theta} = [A]^\bot_{\theta}$.

Let $A \in L(\PA^2)$. We say that $A$ is \emph{valid} if $[A]_\theta = \true$ for all $\theta \in \setE(A)$. Assume $A$ is closed. Then $\emptyset \in \setE(A)$, and we say that $A$ is \emph{true} if $[A]_{\emptyset} = \true$.

\subsection{A game interpretation for $\PA^2$}
In this subsection we define a game interpretation $\Gamma \in L(\PA^2) \mapsto \interpret{\Gamma} \in \setG$ for second order sequents into set-theoretical games, without parameters: parametric games are used only as an intermediate step in the definition. We will prove that this interpretation is sound, complete and effective. 

We take as set $\Par$ of parameters the set of atomic formulas $X(\vec{c_n}), X^\bot(\vec{c_n})$ of $L(\PA^2)$ having head some variable and arguments a list of constants: $\Par$ is equipped with some map $(.)^\bot$, involutory and without fixed point. 
An atomic formula $X(\vec{c_n}) \in L(\PA^2)$ is interpreted by the parametric game $\END(a)$, where the parameter $a$ is the formula $X(\vec{c_n})$ itself. 


We interpret a $k$-ary variable $X$ by the self-dual set of parameters $\Var(X) = \{X(\vec{c_h}), X^\bot(\vec{c_h}) \rceil | \vec{h} \in \Nat^k\}$. We interpret $\vec{X} = X_1, \ldots, X_n$ by the self-dual set $\Var(\vec{X}) = \Var(X_1) \cup \ldots \cup \Var(X_n)$.

We interpret formulas with head symbol $\vee$, $\wedge$ as binary disjunctions, conjunctions with index set $\nth(\{1,2 \})$. We interpret formulas with first order head symbol $\exists$, $\forall$ as  disjunctions, conjunctions with index set $\nth(\Nat)$. $\nth(n)$ denotes the move selecting the immediate sub-formula number $n$, for $n \in \Nat$. We assume that $\pi(.,1)$ is the primitive recursive inverse of $\nth(.)$, that is, that if $i = \nth(n)$ then $n = \pi(i,1)$. 

Given $G_1, G_2 \in \setP$ we define $G_1 \vee G_2 \in \setP$ as $\vee_{i \in \nth(\{1,2\})} G_{\pi(i,1)}$ and $G_1 \wedge G_2$ as $\wedge_{i \in \{1,2\}} G_{\pi(i,1)}$. Both $G_1 \vee G_2$ and $G_1 \wedge G_2$ have immediate sub-games $G_1$, $G_2$. The move selecting the sub-game $G_i$ is, as we anticipated, the move $m = \nth(i)$, because $\pi(m,1) = \pi(\nth(i),1) = i$, therefore $G_{\pi(m,1)} = G_{i}$.


We interpret first order quantifiers in the same way. Given $\{G_n | n \in \Nat\} \subseteq \setP$ we define $\exists n \in \Nat. G_n \in \setP$ as $\vee_{i \in \nth(\Nat)} G_{\pi(i,1)}$ and $\forall n \in \Nat. G_n \in \setP$ as $\wedge_{i \in \nth(\Nat)} G_{\pi(i,1)}$. Both $\exists n \in \Nat. G_n$ and $\forall n \in \Nat. G_n$ have immediate sub-games all $G_n$. As in the previous case, the move selecting the sub-game $G_n$ is $m = \nth(n)$.


We define $(G \Rightarrow H) = (G^{\bot} \vee H)$. 
We may now interpret all formulas of $\PA^2$ by parametric games, and all sequents by games without parameters.

\begin{definition}[Game Interpretation of $L(\PA^2)$]
Assume $A \in L(\PA^2)$ is any formula and $\Gamma = A_0, \ldots, A_{n-1}\in L(\PA^2)$ is any sequent and $g \in \{\Player, \Opponent\}$. Let $\vec{t}$ be a list of closed terms and $\vec{h}$ be the list of values of $\vec{t}$ and $\rho$ be any $\interpret{A}$-assignment.
\begin{enumerate}
\item
If $A = X^k_i(\vec{t}), X^k_i(\vec{t})^\bot$ and $a = X^k_i(\vec{c_h}), X^k_i(\vec{c_h})^\bot$ then $\interpret{A} = \END(a)$.
\item
If $A = p(\vec{t}), p(\vec{t})^\bot$ then
$\interpret{A} = \END(g)$, with $g = \Player$ if and only if $[p(\vec{t})] = \true$.
\item
If $A = A_1 \wedge A_2$ then
$\interpret{A} = \interpret{A_1} \wedge \interpret{A_2}$
\item
If $A = \forall x.B$ then
$\interpret{A} = \forall n \in \Nat.\interpret{B[c_{n}/x]}$
\item
If $A = \FORALL X.B$  and $I = \Var(X) \subseteq \Par$ then
$\interpret{A} = \FORALLGAME I.(\interpret{B})$
\item
If $I = \Var(\FV(\Gamma)) \subseteq \Par$ then
$\interpret{\Gamma} = \FORALLGAME I.(\interpret{A_0}, \ldots, \interpret{A_{n-1}}) \in \setG$.
\item
If $A$ is disjunctive then
$\interpret{A} = \interpret{A^\bot}^\bot$
\end{enumerate}
If $ \rho(\interpret{A}) $ is a determined game, we set $\interpret{A}_\rho = \vince(\rho(\interpret{A}))$.
\end{definition}

In order to establish a correspondence between standard semantics and game semantics of $\PA^2$, we need a correspondence between assignments on $A \in L(\PA^2)$ (maps with codomain $\Bool$) and on $\interpret{A} \in \setP$ (maps with codomain $\{\Player, \Opponent\}$). The correspondence replaces $\Player$ with $\true$ and $\Opponent$ with $\false$. Let $G = \interpret{A}$ and $\theta \in \setE(A)$. We say that $\rho \in \setE(\interpret{A})$ corresponds to $\theta$ if we have: $\rho(a) = \Player$ if and only if $\theta(X)(\vec{n}) = \true$ for all $X \in \FV(A)$ $k$-ary variable, for all $\vec{n} \in \Nat^k$ and for $a = X(\vec{c_n}) \in \Par$. For any $\theta \in \setE(A)$ we may define some $\rho \in \setE(\interpret{A})$ corresponding to it, and for any $\rho \in \setE(\interpret{A})$ we may define some $\theta \in \setE(A)$ corresponding to it.

\section{Effective Soundness and Completeness for second order game semantics}
\label{section-completeness}
In this section we first define exhaustive strategies for $G = \FORALLGAME I.\Gamma$ when $\FV(G) = \emptyset$, that is, when $\FV(\Gamma) \subseteq I$. Then, using exhaustive strategies, we prove that $\FORALLGAME I.\Gamma$ is a sound, complete and effective interpretation of second order quantifier $\FORALL$. Eventually we prove our main result: our second order game semantics is sound, complete and effective for the language $L(\PA^2)$.

The notion of exhaustive strategy is built on the top of the notions of counter-strategies and of non-repeating strategies.

\subsection{Counter-strategies and of non-repeating strategies}
Assume $\Gamma = G_0, \ldots, G_{n-1}$ and $G = \FORALLGAME I.\Gamma$. Then any play $p\in G_{\infty}$ defines a strategy $\tau_i$ which $\Opponent$ follows in the local plays on $G_i$ in $p$. $\tau_i$ consists of all local plays on $G_i$ in $p$: we call $\tau_i$ the \emph{$p$-counter-strategy on $G_i$}. 
 
\begin{definition}[Counter-strategies and non-repeating strategies]
\label{definition-counter-strategy}
\mbox{ }
Let $G = \FORALLGAME I.\Gamma$, $p \in |G|_{\le \infty}$, $m \in \setM$ and $\sigma:G$. 
\begin{enumerate}
\item
The \emph{$p$-counter-strategy on $G_i$} is $\tau_i = \{q \in |G_i| | q \prec_i p\}$. 
\item
$m$ is \emph{not repeated} in $p$ from $G_i$ if for all indexes $j,k$ of local positions of $p$, if $i \vdash_m j$ and $i \vdash_m k$ then $j = k$ \emph{($\Player$ never repeats the move $m$ from $G_i$)}. 
\item
$p$ is \emph{non-repeating} if all $m \in \setM$ are not repeated in $p$ from any $G_i$.
\item
$\sigma$ is \emph{non-repeating}  if all $p \in \sigma$ are non-repeating. 
\end{enumerate}
\end{definition}

The player $\Player$ may repeat the same move $m$ from $G_i$ if he hopes that the second time his opponent will react differently. However, to repeat a move is in principle superfluous: we will prove that all $\Player$-winning strategies may be made non-repeating. Given a non-repeating play, we prove that any infinite play whose proper prefixes are local plays is a local play.

\begin{lemma}[Non-repeating plays]
\label{lemma-non-repeating}
Assume $G = \FORALLGAME I.\Gamma$ and $p \in |G|_{\le \infty}$ is non-repeating and $G_i$ is a local position of $G$. Let $q \in |G_i|_\infty$ and $\tau_i = \{q \in |G_i| | q \prec_i p\}$ the $p$-counter-strategy on $G_i$.
\begin{enumerate}
\item
If $r \prec_i p$ for all $r<q$, then $q \prec_i p$.
\item
If $q \in (\tau_i)_{\infty}$ then  $q \prec_i p$. 
\end{enumerate}
\end{lemma}

\begin{proof}
Assume $p$ is non-repeating in $G$.
\begin{enumerate}
\item 
By induction on the local play $r = \langle m_0, \ldots, m_{n-1} \rangle \prec_i p$, we may prove that $i \vdash_r j$ and $i \vdash_r k$ imply $j = k$, and that there is a unique sequence $i = i_0 < \ldots < i_n = j = k$ such that $i_0 \vdash_{m_0} i_1$, \ldots, $i_{n-1} \vdash_{m_{n-1}} i_{n}$.  As a corollary, any infinite play $q = \langle m_0, \ldots, m_{n-1}, \ldots \rangle \in |G_i|_\infty$ such that $r \prec_i p$ for all $r < q$ defines a unique infinite chain $i_0 \vdash_{m_0} i_1$, \ldots, $i_{n-1} \vdash_{m_{n-1}} i_{n}$, \ldots. By definition of $\prec_i$ we conclude that $q \prec_i p$.
\item
If $q \in (\tau_i)_{\infty}$ then for all $r < q$ we have $r \prec_i p$ by definition of $\tau_i$. By point $1$ above we conclude $q \prec_i p$. 
\end{enumerate}
\end{proof}

Without the assumption assumption ``$p$ non-repeating'' the result above fails.

\subsection{The exhaustive strategies for the game $\FORALLGAME I.\Gamma$}
 Let $G = \FORALLGAME I.\Gamma$, with $\Gamma = G_0, \ldots, G_{n-1} \in \setP^n$ and $\FV(G) = \emptyset$. 
In this sub-section we define a set of strategies for $G$ we call exhaustive strategies. Exhaustive strategies are the worst-case among the $\Player$-winning strategies: the are $\Player$-winning whenever some $\Player$-winning strategy exists, they are slow and clumsy and they just try to win in all possible ways. In the next sub-section we will define one particular exhaustive strategy, $\EXH$, which is primitive recursive w.r.t. the finitary parts of $G_0, \ldots, G_{n-1}$. The existence of $\EXH$ will be essential in order to prove soundness and completeness for our game semantics in an effective way. 

We first define the set of plays in $G$ in which we cannot move $\EM$, and the $G$-assignments for them which are the worst possible for player $\Player$. 

\begin{definition}[The worst possible assignments]
\label{definition-worst}
Let $G = \FORALLGAME I. \Gamma$, with $\Gamma = G_0, \ldots, G_{n-1}$ and $\FV(G) = \emptyset$. 
\begin{enumerate}
\item A (possibly infinite) play $p \in |G|_{\le \infty}$  is $\EM$-forbidding if we cannot move $\EM$ in any finite prefix of $p$: if $a \in I$ and $G_i = \END(a)$ for some $i \in \Nat$ and some local position $G_i$, then $G_j \not = \END(a^\bot)$ for all $j \in \Nat$ and all local positions $G_i$. 

\item
Assume $p \in |G|_{\le \infty}$ is $\EM$-forbidding. Then the \emph{worst-case $p$-assignment} $\eta:\FV(\Gamma) \rightarrow \{\Player, \Opponent\}$ for a (possibly infinite) play $p \in |G|_{\le \infty}$ is defined as follows, for all $a \in I$. Assume $\{G_i | i \in \Nat\}$ is the set of local positions of $p$.
\begin{enumerate}
\item
We set $\eta(a) = \Opponent$ and $\eta(a^\bot) = \Player$ if $G_i = \END(a)$ for some $i \in \Nat$. 
\item
If neither $a$ nor $a^\bot$ have indexes in $p$, we arbitrarily decide for each pair $\{a, a^\bot\}$ whether we assign $\eta(a) = \Player$ and $\eta(a) = \Opponent$, or the other way round.
\end{enumerate}
\end{enumerate}
\end{definition}

The definition of worst-case $p$-assignment is correct because we assumed that if $G_i = \END(a)$ for some $i \in \Nat$, then $G_j \not = \END(a)$ for all $j \in \Nat$. By construction we have $\eta(a)^\bot = \eta(a^\bot)$ for all $a \in \FV(\Gamma)$: $\eta$ is an assignment. By construction we have $\eta(a) = \Opponent$ for all $\END(a)$ which are local positions of $p$. The worst-case $p$-assignment is an arbitrary choice for a $\FV(\Gamma)$-assignment in which $\Opponent$ wins all local positions which are ``generic'' games, i.e., of the form some $\END(a)$.
 


We informally outline the definition of ``$\sigma$ is an exhaustive strategy''. 
$\sigma$ is exhaustive if $\sigma$ is non-repeating and $\sigma$ ``tries all possible ways for $\Player$ to win''. This latter request is expressed as follows: for all maximal plays $p \in \sigma$ 
\begin{enumerate}
\item
From any local position $G_i$ of $p$, if $\Opponent$ moves first and there are moves from $G_i$ then eventually $\Player$ wins or $p$ includes one local play with a move from $G_i$. If $\Player$ moves first then eventually $\Player$ wins or $p $ includes one local play for each possible move from $G_i$.
\item
If some move $\EM$ is possible, , then eventually $p$ contains one winning move and $p$ is terminated and won by $\Player$. 
\item
If some winning $\STOP$ move is possible, that is, if $\Player$ wins in some finite local play $q$ of $p$, then eventually $p$ contains one winning move and the play is terminated and won by $\Player$. 
\end{enumerate}
We express conditions $1,2,3$ by taking the contrapositive. We ask that for all plays $p$ which are terminated by some $\DROP$ move or are infinite we have: \emph{(1)} for any local play $q$, $p$ includes one extension by $\Opponent$ or all possible extensions by $\Player$; \emph{(2)}  no $\EM$-move is possible; \emph{(3)}  in all local plays of $p$, $\Opponent$ always wins for some assignment the local play if the local play is terminated. The formal definition of an exhaustive strategy runs as follows.



\begin{definition}[Exhaustive strategies]
\label{definition-exhaustive}
Let $G = \FORALLGAME I.\Gamma$ and $\FV(G) = \emptyset$. Then $\sigma:G$ is \emph{exhaustive} if 
$\sigma$ is non-repeating, and for all $p \in \sigma$ ending by $\DROP$ or infinite, if $\eta$ is the worst-case $p$-assignment, $i<n$ and $\tau_i = \{q \in |G_i| | q \prec_i p\}$ is the $p$-counter-strategy in $G_i$, then 
\begin{enumerate}
\item
$\tau_i$ is $\Opponent$-total for $G_i$

\item
$p$ is $\EM$-forbidding. 

\item
$\tau_i$ is $\Opponent$-partially winning for $\eta(G_i)$.

\end{enumerate}
\end{definition}

 Remark that \ref{definition-exhaustive}.3 $\implies$ \ref{definition-exhaustive}.1. Indeed, if $\tau_i$ is $\Opponent$-partially winning for $\eta(G_i)$ then it is $\Opponent$-total for $\eta(G_i)$ and it is $\Opponent$-total for $G_i$, because the nodes with children of $G_i$ and $\eta(G_i)$ are the same. \ref{definition-exhaustive}.1 is logically superfluous, but we inserted it to emphasize that each $\tau_i$ is $\Opponent$-total for $G_i$. We prove the existence some exhaustive strategy by defining the exhaustive strategy $\EXH$. $\EXH$ is primitive recursive in $G_0, \ldots, G_{n-1}$, and it is defined by precising in which order $\Player$ selects the next available non-repeating move.

\subsubsection{A sketch of the exhaustive strategy $\EXH$} 
$\EXH$ considers all local positions $G_i$ of $G$, for $i = 0, \ldots, n-1$, in this order. If $G_i$ is atomic, then $\EXH$ makes a move $\EM(i,j)$ or $\STOP(i)$ whenever this move is possible and it is winning, otherwise $\EXH$ skips the local position $G_i$. If $G_i$ is not atomic, and $\Opponent$ is the first player on $G_i$, and no move from $G_i$ was done before, then $\EXH$ asks $\Opponent$ to make a move from $G_i$, otherwise $\EXH$ skips $G_i$. If $\Player$ is the first player on $G_i$, then $\EXH$ makes the first move from $G_i$ which was not done before, otherwise $\EXH$ skips $G_i$. If $\EXH$ looks through all $G_i$, and finds no winning move, the state of the game is some $G' = \FORALLGAME I.\Gamma, G_n, \ldots, G_{n+h-1}$, with $h \in \Nat$ new local positions, and $2h$ new moves. In the case $h=0$, then no winning move is available, otherwise the play would stop before, and no non-repeated move is available. In this case $\EXH$ plays $\DROP(n)$. In the case $h > 0$, then there are new local positions, and $\EXH$ start again from $G'$. $\EXH$ cyclically repeat all these steps, forever, unless $\EXH$ either wins or loses in finite time.

In order to formally describe $\EXH$, we introduce a strategy state $(i,p,q, \vdash)$ associated to each play $r$ in $\EXH$ of length $2h$ for some $h \in \Nat$. We ask that $i \le p \le q$. $i$ is the index of local position $G_i$ which $\EXH$ will consider next. $G_0, \ldots, G_{p-1}$ is the list of local positions which are considered in the current cycle. $G_r = \FORALLGAME I. (G_0, \ldots, G_{q-1})$ is some sub-game $G_r$ of $G$, which is the result of the $2h$ moves of the play $r$, and we have $q = n + h$. $\vdash$ is the list of triples $(i, m, j)$ of the justification relation of $r$. Any sequence of consecutive states $(0,p,q_0, \vdash_0)$, $(1,p,q_1, \vdash_1)$, \ldots, $(p,p,q_p, \vdash_p)$ is called a \emph{cycle}, and any proper non-empty prefix of it an \emph{incomplete cycle}. 

The initial state of $\EXH$ is equal to $(0,n,n,\emptyset)$, and corresponds to the empty play $\nil$. In the initial state, $\EXH$ is considering $G_0$, the current cycle considers $G_0, \ldots, G_{n-1}$, the current sub-game is $G_\nil = G$, and the justification relation is empty. For $i = 0, \ldots, n-1$, $\EXH$ decides whether to move from $G_i$ or not. Eventually, $i$ becomes $n$ and the new state is $(n,n,n+h,\vdash)$, with $2h$ the total number of moves considered by $\EXH$. If $h=0$ no new local position has been inserted, that is, no move is possible but $\DROP$, and $\EXH$ drops out. Otherwise $\EXH$ assigns the value $0$ to $i$ and the new state is $(0,n+h,n+h,\vdash)$. $\EXH$ produces a sequence of cycles: $(0,n_0,n_0,\vdash_0)$, \ldots, $(0,n_1,n_1,\vdash_1)$, \ldots,$(0,n_2,n_2,\vdash_2), \ldots$, with $\vdash_0 = \emptyset$ and $n = n_0 < n_1 < n_2 < \ldots$. The last cycle may be incomplete. The sequence stops if and when $\EXH$ wins or loses in finite time.

\subsubsection{Definition of $\exh(i,p,q, \vdash)$ and $\EXH$}
We define a primitive recursive family $\exh(i,p,q, \vdash)$ of strategies for any $i \le p \le q$ in $\Nat$, and for any $\vdash$ list of triples over $\Nat$, translating the informal definition we sketched above.

The first move of $\exh(i,p,q,\vdash)$ is given by the first true clause in the following list. There is some: indeed, assume that the first four  clauses ($\REPEAT$, $\DROP$,$\EM$, $\STOP$)  are false. Then $i<p\le q$, and if $G_i$ is atomic then $\SKIP$ holds, if $G_i$ is not atomic then $\TRYALL$ holds. All clauses are primitive recursive in the finitary part of $G_0, \ldots, G_{q-1}$: there is relative a primitive recursive map deciding whether there is some $x \in \{\Player, \Opponent\} \cup \Par$ such that $G_i = \END(x)$ and $G_j = \END(x)^\bot$, and finding the first such $x$, if any exists.

\begin{definition}[The strategies $\exh$]
\label{definition-exh}
\mbox{}
\begin{enumerate}

\item[$(\REPEAT)$] 
If $i=p<q$ then $\exh(i,p,q,\vdash) = \exh(0,q,q,\vdash)$ (\emph{$\exh$ starts a new cycle from $G_0, \ldots, G_{q-1}$}).

\item[$(\DROP)$] 
If $i=p=q$ then $\exh(i,p,q,\vdash) = \{\nil, \langle \DROP(q) \rangle\}$ (\emph{There is no non-repeating extension and $\Player$ drops out}).

\item[$(\EM)$] 
If $i<p$ and $G_i = \END(a)$ for some $a \in I$
and there is a first $j<p$ such that $G_j = \END(a^\bot)$, then $\exh(i,p,q,\vdash) = \{\nil, \langle \EM(i,j) \rangle\}$ (\emph{$\Player$ wins}).  

\item[$(\STOP)$] 
If $i<p$ and $G_i = \END(\Player)$, then $\exh(i,p,q,\vdash) = \{\nil, \langle \STOP(i) \rangle\}$ (\emph{$\Player$ wins}).  

\item[$(\SKIP)$] 
If $i<p$ and $G_i$ is atomic and none of the previous cases applies, then $\exh(i,p,q,\vdash) = \exh(i+1,p,q,\vdash)$, with $i+1 \le p \le q$ (\emph{$\exh$ skips $G_i$ and analyzes $G_{i+1}$}).

\item[$(\TRYALL)$] 
If $i<p$ and $G_i$ is \emph{not} atomic, let $I$ be the set of moves from $G_i$.
\begin{enumerate}
\item
\emph{Assume $\Player = \turno_{G_i}(x)$}. Assume there is some $m \in I$ such that $i \vdash_m k$ is false for all $k<q$. Take the first of them \emph{($m$ is the first move from $G_i$ never done before, if any)}. Then $\exh(i,p,q,\vdash)$ moves $\JUST(i,q)$, then moves $m$ from $G_i$ for $\Player$ and updates $i, q, \vdash$. We set
$$
\exh(i,p,q,\vdash) 
=
\{\nil, \langle \JUST(i,q) \rangle \} 
\ \cup \
\JUST(i,q) @ m @ \exh(i+1,p,q+1, \vdash @ (i,m,q))
$$ 
If no such $m \in I$ exists, then $\exh(i,p,q,\vdash)$ 
skips $G_i$ and analyzes $G_{i+1}$: we set 
$$
\exh(i,p,q,\vdash) 
=
\exh(i+1,p,q,\vdash)
$$

\item
\emph{Assume $\Opponent = \turno_{G_i}(x)$}. Assume there is no $m \in I$ such that $i \vdash_m k$ for any $k < q$. Then $\exh(i,p,q,\vdash)$ first moves $\JUST(i,q)$, then waits for a move $m \in I$ from $\Opponent$ and updates $i, q, \vdash$. We set:
$$
\exh(i,p,q,\vdash)
=
\{\nil, \langle \JUST(i,q) \rangle \} 
\ \cup \ 
\ \bigcup_{m \in I} \
\JUST(i,q) @ m @ \exh(i+1,p,q+1, \vdash @ (i,m,q))
$$
If there is such an $m$, then $\exh(i,p,q,\vdash)$
skips $G_i$ and analyzes $G_{i+1}$: we set
$$
\exh(i,p,q,\vdash) 
= 
\exh(i+1,p,q,\vdash)
$$ 
\end{enumerate}
\end{enumerate}
Eventually we set $\EXH = \exh(0,n,n,\emptyset)$. 
\end{definition}


$\EXH$ starts from the state $(0,n,n,\emptyset)$. After applying $n$ clauses $\EXH$ either wins by $\EM$ or $\STOP$, or reaches the state $(n,n,n+h,\vdash)$, the end of the cycle: for $2h$ times $\EXH$ moves for $\Player$ or considers all possible moves from $\Opponent$. If $h=0$ then $\EXH$ plays $\DROP$ and loses. If $h>0$ then $\EXH$ starts a new cycle from the state $(0,n+h,n+h,\vdash)$. Thus, in any state non corresponding to a terminated play, eventually $\EXH$ moves for $\Player$ or considers all moves from $\Opponent$. 

If $p \in \EXH$ is infinite or terminated by $\DROP$ then any local position $G_i$ of $p$ is analyzed at least once. Indeed, if $p$ is infinite then any local position is analyzed in all cycles after it is introduced. If $p \in \EXH$ is terminated by $\DROP$, then the last cycle adds no local positions, therefore all local positions are analyzed at least in the last cycle.


\begin{lemma}[exhaustive strategies]
\label{lemma-exhaustive}
Let $G = \FORALLGAME I.\Gamma$ and $\FV(G) = \emptyset$. 
Then $\EXH$ is a primitive recursive exhaustive strategy for $G$.
\end{lemma}

\begin{proof}
$\EXH$ is a primitive recursive tree support by definition. We already noticed that eventually $\EXH$ moves for $\Player$ or consider all possible moves from $\Opponent$. Thus, $\EXH$ is a $\Player$-total strategy. $\EXH$ is non-repeating because the clauses $\TRYALL.a$ and $\TRYALL.b$ explicitly forbid a repetition, and the other clauses make some move $\DROP$, $\EM$, $\STOP$ which terminate the play, and therefore is unique. We check the conditions \ref{definition-exhaustive}.1,2,3 of the definition of exhaustive strategy. Assume that $p \in \EXH_{\le \infty}$ is terminated by $\DROP$ or $p$ is infinite, that $i < n$ and $\tau_i = \{q \in |G_i| | q \prec_i p\}$.

\begin{enumerate}
\item
{\em $\tau_i$ is $\Opponent$-total}.
If $G_i$ is atomic then any $G_i$-strategy is $\Opponent$-total, since there are no moves. Assume $G_i$ is not atomic. Let $I$ be the set of moves from $G_i$. Assume $q \in \tau_i$: then $G_j = (G_i)_q$ is some local position of $p$. We already noticed that $\EXH$ analyzes $G_j$ at least once. Assume $\Opponent$ moves from $q$. By clause $\TRYALL.a$, $\EXH$ adds one move $m \in I$ from $G_j$, hence $q @ m \prec_i p$ and $q @ m \in \tau_i$. Assume $\Opponent$ moves from $q$. By clause $\TRYALL.b$ and $\REPEAT$, $\EXH$ moves from $G_j$ as many times as there are moves $m \in I$ from $G_j$ which are new, adding each time a new local position, then executing a new cycle. Thus, for all $m \in I$ we have $q @ m \prec_i p$ and $q @ m \in \tau_i$.

%

\item
\emph{$p$ is $\EM$-forbidding}. Assume $p$ is not: then there are some $G_i = \END(x)$, $G_j = \END(x^\bot)$ for some $x \in \FV(\Gamma)$. Take the first such $j$, and the first $i$ for such $j$. Eventually, either $\Player$ wins in $p$ or $\EXH$ analyzes $G_j$ and by clause $\EM$ moves $\EM(i,j)$. In both cases $\Player$ wins and the play $p$ stops. This contradicts the assumption that $p$ ends with some $\DROP$ move or that $p$ is infinite.

\item 
{\em $\tau_i$ is $\Opponent$-partially winning}.
Let $\eta$ be the worst-case $p$-assignment: since $p$ is $\EM$-forbidding by point $2$ above,
and $\FV(G) = \emptyset$, 
then $\eta(a)  = \Opponent$ for all local positions $\END(a)$ in $p$. Let $\tau_i : G_i$ be the set of $q \in |G_i|$ such that $q \prec_i p$. By point $1$ above, $\tau_i$ is $\Opponent$-total for $G_i$, hence for $\eta(G_i)$, because $\eta(G_i)$ and $G_i$ coincide on all nodes having children. In order to prove that $\tau_i$ is partially $\Opponent$-winning on $\eta(G_i)$, we have to prove that all finite terminated $q \in \tau_i$ are won by $\Opponent$ in $\eta(G_i)$. We have $q \prec_i p$ by definition of $\tau_i$. The last position of a local play $q$ is $G_k = \END(x)$ for some $k \in \Nat$, $x \in \Par \cup \{\Player, \Opponent\}$. By cases on $x$ we prove that $\eta(\END(x)) = \END(\Opponent)$, that is, that $\eta(x) = \Opponent$.
\begin{enumerate}
\item
Let $x = a \in \Par$. Then we have $\eta(a) = \Opponent$ by definition of $\eta$.   
\item
Let $x = \Opponent$. Then $\eta(x) = \eta(\Opponent) = \Opponent$.
\item
Let $x = \Player$. Then we would have $G_k = \END(\Player)$, but we prove that this cannot be. Eventually, either $\Player$ wins $p$, or $\EXH$ analyzes $G_k$, and by clause $\STOP$ moves $\STOP(k)$, $\Player$ wins and $p$ stops. This contradicts the assumption that $p$ ends with some $\DROP$ move or $p$ is infinite.
\end{enumerate}
\end{enumerate}
\end{proof}

\subsection{Effective Soundness and Completeness result for the game $\FORALLGAME I.\Gamma$}
In this subsection we characterize the game $G = \FORALLGAME I.\Gamma$, first in the case $\FV(G) = \emptyset$ and then in general.

\begin{lemma}[the game $\FORALLGAME I.\Gamma$]
\label{lemma-forall}
Let $\Gamma = (G_0, …, G_{n-1}) \in \setP^n$ be a list of $n$ games, $I \subseteq \Par$ a self-dual set of parameter and $G = \FORALLGAME I.\Gamma$. 
Assume $\FV(G) = \emptyset$ (hence $I \subseteq \FV(\Gamma))$. 
These following are equivalent. 

\begin{enumerate}
\item
There is some assignment $\eta:I \rightarrow \{\Player, \Opponent\}$ and for all $i<n$, some $\Opponent$-winning strategy $\tau_i : \eta(G_i)$
\item
There is some $\Opponent$-winning strategy $\tau$ for $\FORALLGAME I.\Gamma$.
\item
There is no $\Player$-winning strategy $\sigma$ for $\FORALLGAME I.\Gamma$.
\item
The primitive recursive exhaustive strategy $\EXH$ for $\FORALLGAME I.\Gamma$ is not $\Player$-winning.
\end{enumerate}

\end{lemma}

\begin{proof}
\begin{itemize}

\item
$(1 \implies  2)$. 
Assume that there is some assignment $\eta : I \rightarrow \{\Player, \Opponent\}$ and for all $i<n$, some $\Opponent$-winning strategy $\tau_i:\eta(G_i)$. We have to define some $\Opponent$-winning strategy $\tau : \FORALLGAME I.\Gamma$. We define $\tau$ as $\{p \in |G| | \forall i<n.\forall  q \prec_i p. q \in \tau_i\}$. $\tau$ is the strategy which follows $\tau_i$ on each local play on $G_i$, for $i = 0, \ldots, n-1$. $\tau$ is an $\Opponent$-strategy and it is $\Opponent$-total because all $\tau_i$ are. We have to prove that $\tau$ is $\Opponent$-winning. For any $p \in \tau$ and any local position $G_j$ of $p$, we have $G_j = (G_i)_q$ for exactly one $i<n$, one $q \in \tau_i$. Since each $\tau_i$ is $\Opponent$-winning, we deduce that $(\tau_i)_q : G_j$ is an $\Opponent$-winning strategy on $\eta(G_j)$. Thus, there is some $\Opponent$-winning strategy for all local positions of $p$. For $\tau$ to be $\Opponent$-winning, we still have to prove that all (possibly infinite) maximal plays $p \in \sigma_{\le \infty}$ are won by $\Opponent$. We argue by cases. Assume $p$ is finite: then either $p$ ends by a move $\DROP(i)$ by $\Player$, or a by a move $\EM(i,j)$ by $\Player$, or with $\STOP(i)$, for some $i,j$. If $\Player$ moves $\DROP(i)$ then $\Opponent$ wins. $\Player$ cannot move $\EM(j,k)$, otherwise there are two local positions $G_j, G_k$ of $p$ such that $G_j = G^\bot_k$, and we proved that $\Opponent$ has some winning strategies for both $\eta(G_j) = \eta(G_k)^\bot$ and $\eta(G_k)$, contradiction. Any finite terminated $q \prec_i p$ is in the $\Opponent$-winning strategy $\tau_i$ and therefore if $\Player$ moves some $\STOP(k)$ then $\eta(G_k) = \END(\Opponent)$. By definition of $\STOP$ move we have $\FV(G_k) \cap I = \emptyset$: from $\FV(G_k) \subseteq \FV(\Gamma) \subseteq I$ we conclude $\FV(G_k) = \emptyset$, hence $G_k = \END(\Opponent)$. Thus, $\Opponent$ wins $p$. Let us assume that $p$ is infinite. Then for all $i = 0, \ldots, n-1$, all infinite $q \prec_i p$ are in $(\tau_i)_\infty$ and therefore are won by $\Opponent$. In all cases, $p$ is won by $\Opponent$.

\item
$(2 \implies  3)$. 
Assume that is some $\Opponent$-winning strategy $\tau$ for $\FORALLGAME I.\Gamma$. Then there is no $\Player$-winning strategy $\sigma$ for $\FORALLGAME I.\Gamma$, otherwise we would produce a contradiction by letting them play together.

\item
$(3 \implies  4)$.
Assume there is no $\Player$-winning strategy $\sigma$ for $\FORALLGAME I.\Gamma$. Then, in particular, the strategy $\EXH$ is not $\Player$-winning for $\FORALLGAME I.\Gamma$. 

\item
$(4 \implies  1)$.
Assume that the strategy $\EXH$ is not $\Player$-winning for $G = \FORALLGAME I.\Gamma$. We have to define some $\FV(\Gamma)$-assignment $\eta$ and some $\Opponent$-winning strategies for $\eta(G_0), \ldots, \eta(G_{n-1})$.
Since $\EXH$ is $\Player$-partially winning but not $\Player$-winning, there is some infinite play $p \in \EXH \cap G_{\Opponent}$. Let us choose any. We take for $\eta$ the worst-case $p$-assignment. Let $\tau_i : G_i$ be the $p$-counter-strategy $\{q \in |G_i| | q \prec_i p\}$. $\EXH$ is exhaustive by Lemma \ref{lemma-exhaustive}, hence, by definition of exhaustive, $\tau_i$ is $\Opponent$-partially winning on $\eta(G_i)$. To prove that $\tau_i$ is $\Opponent$-winning, we have to prove that all infinite plays $q \in (\tau_i)_{\infty}$ are in $\eta(G_i)_{\Opponent}$. By $\EXH$ non-repeating and Lemma \ref{lemma-non-repeating}.2 we deduce that $q \prec_i p$. By definition of $p \in G_{\Opponent}$, all $q \prec_i p$ are in $(G_i)_{\Opponent}$, and by definition of $\eta(G_i)_{\Opponent}$ we have $q \in (G_i)_{\Opponent} = \eta(G_i)_{\Opponent}$.

\end{itemize}
\end{proof}

We may now prove that $\EXH$ is some primitive recursive $\Player$-winning strategy for $G = \FORALLGAME I.\Gamma$, whenever $\FV(G) = \emptyset$ and some $\Player$-winning strategy exists for $\FORALLGAME I.\Gamma$.

\begin{theorem}
[Effective Soundness and Completeness for $\FORALLGAME I.\Gamma \in \setG$] 
\label{theorem-forall}
Let $\Gamma = G_0, \ldots, G_{n-1}$ be $n$ parametric games and $G = \FORALLGAME I.\Gamma$ and $\FV(G) = \emptyset$.
\begin{enumerate}
\item
The game $G$ is determined.
\item
Assume all $\eta(G_0), \ldots, \eta(G_{n-1})$ are determined, for all $I$-assignment $\eta$. Then the following are equivalent.
\begin{enumerate}
\item
There is some $\Player$-winning strategy $\sigma:G$. 
\item
$\EXH$ is some primitive recursive $\Player$-winning strategy for $G$.
\item
For any $I$-assignment $\eta$ there is some $i<n$ and some $\Player$-winning strategy $\sigma_i : \eta(G_i)$.
\end{enumerate}
\end{enumerate}
\end{theorem} 

\begin{proof}
\begin{enumerate}
\item
Assume there is some $\Player$-winning strategy for $G$: then $G$ is determined. Assume there is no $\Player$-winning strategy for $G$: then by Lemma \ref{lemma-forall}, point $2\implies 3$, there is some $\Opponent$-winning strategy for $G$. Also in this case $G$ is determined.

\item 
If we take the negation of points $1,3,4$ of Theorem \ref{lemma-forall}, we obtain that the following are equivalent: ``there exists some $\Player$-winning strategy for $G$'' (point $2.a$), ``$\EXH$ is some $\Player$-winning strategy for $G$'' (point $2.b$), and ``for any $\FV(\Gamma)$-assignment $\eta$ there is some $i<n$ and no $\Opponent$-winning strategy $\tau_i:\eta(G_i)$''. Since we assumed that all $\eta(G_i)$ are determined, this latter is equivalent to: `for any $\FV(\Gamma)$-assignment $\eta$ there is some $i<n$ and some $\Player$-winning strategy $\sigma_i:\eta(G_i)$''. This is point $2.c$ of the thesis.
\end{enumerate}
\end{proof}

We may prove a similar result for any game of the form $\rho(\FORALLGAME I.\Gamma)$.

\begin{corollary}
[Effective Soundness and Completeness for $\FORALLGAME I.\Gamma$] 
\label{corollary-forall}
Let $\Gamma = G_0, \ldots, G_{n-1}$ be $n$ parametric games and $G = \FORALLGAME I.\Gamma$ and $\rho \in \setE(G)$ be any parameter assignment.
\begin{enumerate}
\item
The game $\rho(G)$ is determined.
\item
If for all $I$-assignment $\eta$ the games $(\rho,\eta)(G_0), \ldots, (\rho,\eta)(G_{n-1})$ are determined,hen the following are equivalent.
\begin{enumerate}
\item
There is some $\Player$-winning strategy $\sigma:\rho(G)$. 
\item
$\EXH$ is some primitive recursive $\Player$-winning strategy for $\rho(G)$.
\item
For any $I$-assignment $\eta$ there is some $i<n$ and some $\Player$-winning strategy $\sigma_i : (\rho,\eta)(G_i)$.
\end{enumerate}
\end{enumerate}
\end{corollary} 

\begin{proof}
Let $\rho' = $ the restriction of $\rho$ to $\FV(G) = \FV(\Gamma) \setminus I$. Then $\rho(G) = \rho'(G) = \forall I. \rho'(\Gamma)$ and for all $i<n$ we have $\eta(\rho'(G_i)) = $ (by def.) $(\eta,\rho')(G_i) =$ (since $\rho', \eta$ have disjoint domains) $(\rho', \eta)(G_i) =$ (on $(\FV(G_i) \setminus I)$, the maps $\rho, \rho'$ coincide) $ (\rho, \eta)(G_i)$. Besides, $\FV(\rho'(G)) = \FV(\rho(G)) = \emptyset$. We conclude our thesis by Thm. \ref{theorem-forall} applied to $\rho'(G)$. 
\end{proof}

We are now ready to prove our main theorem for $L(\PA^2)$.

\begin{theorem}[Effective Soundness and Completeness for $L(\PA^2)$]
\label{theorem-main}
Assume $A \in L(\PA^2)$ is a $1$-closed formula and $\Gamma \in L(\PA^2)$ is a $1$-closed sequent. Let $\theta \in \setE(A)$ and $\rho \in \setE(\interpret{A})$.
\begin{enumerate}
\item
$\rho(\interpret{A})$ is a determined game.

\item
If $\rho$, $\theta$ are corresponding assignments, then $[A]_\theta = \true$ if and only if $\interpret{A}_\rho = \Player$.

\item
$\interpret{A} \in \setP$ is a valid parametric game if and only if $A \in L(\PA^2)$ is a valid second order formula.

\item
if $\Gamma$ is valid, then there is some primitive recursive winning strategy for $\interpret{\Gamma} \in \setG$.

\end{enumerate}
\end{theorem}

\begin{proof}
Assume $\rho$ is an $\interpret{A}$-assignment.
\begin{enumerate}
\item
By induction on the definition of $\rho(\interpret{A})$. If $A$ is an atomic formula then $\rho(\interpret{A}) = \END(g)$ for some $g \in \{\Player,\Opponent\}$ is determined. If $A = A_1 \wedge A_2$ and each $\rho(\interpret{A_i})$ is determined then $\rho(\interpret{A}) = \rho(\interpret{A_1}) \wedge \rho(\interpret{A_2})$ is $\Player$-winning if both $\rho(\interpret{A_1})$, $\rho(\interpret{A_2})$ are $\Player$-winning, otherwise $\Opponent$ wins playing some move $\nth(i)$ such that $\rho(\interpret{A_i})$ is $\Opponent$-winning. A similar reasoning applies for $A = \forall x.B$. Assume $A = \FORALL X.B$: then $\rho(\interpret{A})$ is determined by Corollary \ref{corollary-forall}.1. If $A$ is disjunctive then $\rho(\interpret{A^\bot})$ is determined, therefore 
$\interpret{A}_\rho = \rho(\interpret{A^\bot})^\bot$ is determined.


\item
By induction on the definition of $\interpret{A}$. Assume $A = \FORALL X.B$.
Then by Corollary \ref{corollary-forall}.2 we have: $\interpret{A}_{\rho} = \Player$ if and only if for any $\Var(X)$-assignment $\eta$ we have $\interpret{B}_{(\rho,\eta)} = \Player$ if and only if (by induction hypothesis on $\interpret{B}$) for any $\{X\}$-assignment $\psi$ we have $\interpret{B}_{(\theta,\psi)} = \true$ if and only if (by definition of $\interpret{A}_{\theta}$) we have $\interpret{A}_{\theta} = \true$. All other cases follows immediately by induction hypothesis.

\item
For any $\theta \in \setE(A)$ we may define some $\rho \in \setE(\interpret{A})$ corresponding to it, and conversely. The condition $\rho(a^\bot) = \rho(a)^\bot$ follows from $\theta(X^\bot) = \theta(X)^\bot$, and conversely. Thus, by point $2$ above, if $\interpret{A}_\rho = \Player$ for all $\rho \in \setE(\interpret{A})$, then $[A]_\theta = \true$ for all $\theta \in \setE(A)$, and conversely.

\item
By point $3$ above there is some $\Player$-winning strategy for $\interpret{\Gamma}$. By Theorem \ref{theorem-forall} and $\FV(\interpret{\Gamma}) = \emptyset$ the strategy may be chosen primitive recursive.
\end{enumerate}
\end{proof}

\section{Comparing with previous works}
\label{section-conclusion}
As we explained in the introduction, the game semantics more similar to our one is the game semantics by De Lataillade (\cite{JD08a, JD08b}), having no winning conditions for infinite plays, and used for characterizing type isomorphisms for a functional language, system $F$. Our game semantics has different definitions and a different goal. We do have the notion of nodes representing variables, as in as in \cite{JD08b}, but we do \emph{not} have in our game semantics a rule corresponding to second order elimination rule $\FORALL X.A \implies A[P/X]$, for $P$ predicate, as in \cite{JD08b}, \S 3.2. Indeed, our long-term goal is to decompose second order elimination rule into simpler rules, and to make a proof theoretical analysis of it. This is why we cannot assume second order elimination from the start, but we have instead to prove that our game semantics is sound for it. 

Our game semantics is built over the work of Lorenzen. Lorenzen introduced the idea of a winning move which matches affirmed and negated occurrence of an atomic formula (\cite{Lorenzen}, \S 1, Def. (D10)), similar to our $\EM$ move matching $X(\vec{t})$ with $X^\bot(\vec{t})$, for $X$ bound predicate variable. Another contribution of Lorenzen is the idea of backtracking/justification move, which he called ``sequence of references'' (\cite{Lorenzen}, \S 1). We also based our work over the works of Coquand (\cite{Coquand-1991, Coquand-1995}). Coquand, with his former ph.d. student Herbelin (\cite{HerbelinPhD}), introduced a synthetic combinatorial definition of game with backtracking/justification moves, proved that game semantics with backtracking is sound and complete for first order arithmetic, and defined an effective cut-elimination procedure for it. Game semantics was adapted to an interpretation of first order functional programming languages by Hyland and Ong (\cite{Hyland-Ong}). 

Second order game semantics of De Lataillade (\cite{JD08a, JD08b}) is build over the work by Hyland and Ong, adding a notion of ``generic'' node representing second order variables, and a second order elimination rule, which allow a player to replace a node representing a variable by any tree representing a game. We retained the first feature of his semantics and we dropped the second one, for the reasons we explained. Another difference is: Lataillade only considers strategies which are \emph{uniform}, that is, moving independently from the substitution $[P/X]$ in $\FORALL X.A \implies A[P/X]$ (\cite{JD08a}, Def. 21). Since we have no game rule expressing the rule $\FORALL X.A \implies A[P/X]$, we do not explicitly require uniformity. We could say that \emph{all strategies are uniform} in our second order game semantics, in the sense they do not depend on the boolean functions assigned to variable predicates by the ``judge of the play''. 

Another essential difference is that we added the winning conditions for infinite plays, which are not trivial. As we already pointed out, by Tarski's undefinability theorem, the family of winning conditions for a sound and complete interpretation of truth for $L(\PA^2)$ cannot be defined in $L(\PA^2)$ itself. In the papers of Hyland, Ong and Lataillade, winning conditions are not considered. These papers interpret a programming language, and they do not need winning conditions for this purpose.

\section{Acknowledgments}
We thank Federico Asperti for checking an early version of this paper and for suggesting several improvements. We thank Erik Krabbe, Shahid Rahman, Helge Ruckert for helpful comments and Gabriel Sandu for suggesting many interesting related works to read. We thank Silvia Steila for quotations about Set Theory. 

\bibliographystyle{asl}
\bibliography{Berardi-A-Sound-Complete-and-Effective-Game-Semantics-for-Second-Order-Arithmetic}

\end{document}